\documentclass{IEEEtran}
\usepackage{subfigure}
\usepackage{amssymb}
\usepackage{url}
\usepackage{color}
\usepackage{setspace}
\usepackage{amsmath}
\usepackage{algorithm}
\usepackage{algpseudocode}
\usepackage{graphics}
\usepackage{epsfig}
\usepackage{cite}
\usepackage{array}
\newcolumntype{C}[1]{>{\centering\let\newline\\\arraybackslash\hspace{0pt}}m{#1}}
\newcolumntype{L}[1]{>{\raggedright\let\newline\\\arraybackslash\hspace{0pt}}m{#1}}
\newcolumntype{R}[1]{>{\raggedleft\let\newline\\\arraybackslash\hspace{0pt}}m{#1}}
\usepackage{multirow}

\newcommand{\defeq}{\stackrel{def}{=}}

\newcommand{\bm}[1]{\boldsymbol{#1}}
\newcommand{\mb}[1]{\mathbf{#1}}
\newcommand{\tr}{\mathrm{tr}}
\newcommand{\var}{\mathrm{Var}}
\newcommand{\MSE}{\mathrm{MSE}}
\newcommand{\AS}{\mathrm{AS}}
\newcommand{\SNR}{{\alpha}}
\newcommand{\NMSE}{\mathrm{NMSE}}
\newcommand{\Diag}{\mathrm{Diag}}

\newtheorem{lemma}{Lemma}

\newtheorem{remark}{Remark}

\newtheorem{corollary}{Corollary}
\newtheorem{definition}{Definition}

\begin{document}

\title{Reduced-Rank Channel Estimation for Large-Scale MIMO Systems\thanks{This work
was supported in part by Taiwan's Ministry of Science and Technology under Grant NSC 102-2221-E-009-016-MY3 and
by MediaTek under Grant 101C125. The material in this paper was presented in part at the 2013 IEEE Globecom Workshops.}
\author{Ko-Feng Chen, Yen-Cheng Liu, and Yu T. Su$^\dag$\thanks{$\dag$Correspondence addressee.}\\}
\thanks{K.-F. Chen is with MediaTek Inc., Hsinchu, Taiwan (email: ko-feng.chen@mediatek.com).
Y.-C. Liu and Y. T. Su are with the Institute of Communications Engineering, National
Chiao Tung University, Hsinchu, Taiwan (email: ycliu@ieee.org; ytsu@nctu.edu.tw). }}

\DeclareGraphicsExtensions{.pdf}
\graphicspath{{pdfFigs/}}
\maketitle \thispagestyle{empty}

\begin{abstract}
We present two reduced-rank channel estimators for large-scale multiple-input,
multiple-output (MIMO) systems and analyze their mean square error (MSE) performance.
Taking advantage of the channel's transform domain sparseness, the estimators
yield outstanding performance and may offer additional mean angle-of-arrival
(AoA) information. It is shown that, for the estimators to be effective, one has
to select a proper predetermined unitary basis (transform) and be able to determine
the dominant channel rank and the associated subspace. Further MSE analysis reveals
the relations among the array size, channel rank, signal-to-noise ratio (SNR),
and the estimators' performance. It provides rationales for the proposed rank
determination and mean AoA estimation algorithms as well.

An angle alignment operation included in one of our channel models is proved to be
effective in further reducing the required rank, shifting the dominant basis vectors'
range (channel spectrum) and improving the estimator's performance when a suitable basis
is used. We also draw insightful analogies among MIMO channel modeling, transform coding,
parallel spatial search, and receive beamforming. Computer experiment results are provided
to examine the numerical effects of various estimator parameters and the advantages
of the proposed channel estimators and rank determination method.
\end{abstract}

\begin{IEEEkeywords}
   Channel estimation, channel rank, channel spatial correlation, massive MIMO, transform-domain approach.
\end{IEEEkeywords}


\section{Introduction}
%

We consider a cellular mobile network in which each base station (BS) is equipped
with a large-scale antenna array whose size $M$ is much greater than the number of
single-antenna mobile users it serves. Such a large-scale (distributed) multiple-input,
multiple-output (MIMO) system or a massive MIMO system for short has the potentiality
of achieving transmission rate much higher than those offered by current cellular
systems with enhanced reliability and drastically improved power efficiency. It takes
advantage of the so-called channel-hardening effect \cite{Scaling} which implies that
the channel vectors seen by different users tend to be mutually orthogonal and
frequency-independent \cite{Measurement2}. As a result, linear receivers are near-optimal
in the uplink and simple multiuser (MU) precoders are sufficient to guarantee satisfactory
downlink performance. We consider a single cell within such a network whose BS array
size $M \gg K$, $K$ being the number of single antenna mobile users within its coverage
range, and refer to this system as a distributed massive MIMO system.



To estimate a massive MIMO channel, we can employ the conventional pilot-assisted
channel estimators (CEs) such as least-squares (LS) \cite{Opt_pilot} or minimum mean
square error (MMSE) estimators \cite{MMSE_CE}. Taking advantage of the massive MIMO
channels' spatial sparseness, CEs based on compressive sensing (CS) have been proposed
\cite{CE_liter1}. However, the complexity of CS recovery is high and CS-based methods
often call for the use of randomized pilots to produce a sensing matrix which satisfies
the restricted isometry property (RIP). They also assume that the channel has limited
paths with separable angles of arrivals (AoAs). In fact, there are not many articles
consider MIMO channels with a cluster of stochastic AoAs around a mean angle. To
circumvent such a limitation, a complicated Bayesian learning approach with a Gaussian-mixture
channel model, was developed in \cite{CE_liter3}. Assuming the channel can be represented
by a nonnegative combination of limited number of elements from the so-called atomic set,
\cite{CE_liter5} presented a denoising-based CE.

By exploiting the transform domain sparseness, \cite{Model_based,GCommwork,CE_liter2}
presented transform domain CEs that use only significant components in the transform
domain to reduce the noise effect and computing complexity. While \cite{Model_based} and
\cite{GCommwork} explored various transforms, \cite{CE_liter2} and the references
therein focused on the discrete Fourier transform (DFT) domain.

In this paper, we extend and analyze the performance of the transform-domain approach
that exploits the spatial correlation-induced channel rank reduction. We show that with
a judicial choice of the unitary transform and a proper selection of the dominant
subspace, the resulting CE gives excellent performance as it involves only the critical
channel dimensions which is much smaller than that required by the conventional methods
\cite{Opt_pilot}. In other words, the efficiency of a transform-domain CE depends not
only on the unitary basis (transform) used but on how one selects a proper subset of the
basis vectors which span the dominant channel subspace. Furthermore, we verify that
improved performance is attainable if an additional angle alignment operation is in place
so that the receive array points to the mean AoA.

We discuss the effect of the basis (transform) used has on the CE's MSE performance and
explain the near-optimality of the type-2 discrete cosine transform (DCT-2) basis. Our
analysis enables us to find the relations between some channel parameters such as channel
rank, array size, signal-to-noise ratio (SNR), and the MSE performance for a given basis.
It also leads to a robust rank determination algorithm. The proposed CE, as will be seen,
can provide additional information about the mean AoA as well. This information is very
useful in designing a downlink precoder. Due to space limitation, we focus on the scenario
where each mobile station (MS) served use a single omnidirectional antenna and there is
only one dominant uplink arrival cluster for every MS-BS link.

The rest of this paper is organized as follows. In Section \ref{section:model}, a brief
review on the system and channel models is given and the advantage of using one of the
models is explained. Section \ref{sec:SSFC} presents the proposed CE and its MSE analysis,
discusses possible implementation complexity reduction, and provides a different perspective
on our approach. In Section \ref{sec:basis} we study the basis selection issue from the
perspective of transform coding and analyze various system and channel parameters' impacts
on the estimator's MSE performance in Section \ref{sec:RR_MO}. We present a rank and the
corresponding subspace indices determination algorithm in Section \ref{subsec:MOD}.
Simulation results are provided in Section \ref{section:simulation} to validate the
superiority of our CE and rank determination algorithm. We summarize our main contributions
in Section \ref{section:conclusion}.

\textit{Notation:}
$\mathrm{vec}(\cdot)$ represents the operator that stacks columns of the enclosed matrix,
whereas $\mathbb{E}\{\cdot\}$, $\|\cdot\|$, $\|\cdot\|_{2}$, and $\|\cdot\|_{F}$ denote
respectively the expectation, vector $\ell_2$-norm, matrix spectral norm, and Frobenius
norm of the enclosed item. We denote by $\mb{I}_n$, $\mb{1}_n$, and $\mb{0}_n$, the $n
\times n$ identity matrix and $n$-dimensional all-one and all-zero column vectors.
$\Diag(\cdot)$ translates the enclosed vector or elements into a diagonal matrix with
the diagonal terms being the enclosed items.

\section{System Model}
\label{section:model} Throughout this paper, we consider a single-cell massive MU-MIMO
system having an $M$-antenna BS and $K$ single-antenna MSs, where $M\gg K$.
\subsection{Operation Scenario and Assumptions}
We assume a narrowband (NB) communication environment in which a transmitted
signal suffers from both large- and small-scale fading. We denote $K$ MS-BS
link ranges by $d_k$ and assume that each uplink packet place its pilot of
length $T$ at the same time location. Given the initial frequency and time
synchronization has been achieved, we express the corresponding $M \times T$
noise variance normalized received samples at the BS, arranged in matrix and
frame-synchronized form, $\mb{Y}=[y_{ij}]$ as
\begin{IEEEeqnarray}{rCl}
 \mb{Y}=\sum_{k=1}^K\sqrt{\beta_k}\mb{h}_k\mb{p}_k^H+\mb{N}
    =\mb{H}\mb{D}_{\boldsymbol{\beta}}^{\frac{1}{2}}\mb{P}+\mb{N},
\label{eq:sys_mod}
\end{IEEEeqnarray}
where $\mb{H} = [\mb{h}_1, \cdots, \mb{h}_K]\in \mathbb{C}^{M\times K}$ and
$\mb{D}_{\boldsymbol{\beta}}=\text{Diag}({\boldsymbol{\beta}})$ contain respectively
the small-scale fading coefficient (SSFC) vectors $\{\mb{h}_k\}_{k=1}^K$ and
large-scale fading coefficient (LSFC) vector $\bm{\beta}=[\beta_1, \cdots,
\beta_K]^T$ that characterize the $K$ uplink vector channels and $\mb{N}=[n_{ij}]$
is the white Gaussian noise matrix with independent identically distributed (i.i.d.)
elements, $n_{ij}\sim \mathcal{CN}(0,1)$. Each element of the LSFC vector, $\beta_k
= {s_k}d_k^{-a}$, is the product of random variable $s_k$ representing the
shadowing effect and the pathloss $d_k^{-a}$, $a > 2$. Assume that $s_k$'s
are i.i.d. lognormal random variables, i.e., $10\log_{10}(s_k)\sim \mathcal{N}(0,
\sigma_s^2)$. The $K\times T$ matrix $\mb{P}=\left[\mb{p}_1,\cdots,\mb{p}_K \right]^H$,
where $T\geq K$, consists of orthogonal uplink pilot vectors $\{\mb{p}_k\}_{k=1}^K$.
The optimality of using orthogonal pilots was verified in \cite{Opt_pilot}.

A normalized NB uplink SSFC vector can be expressed as
\begin{equation}
 \mb{h}=\frac{1}{\sqrt{LS}}\sum_{\ell=1}^L\sum_{s=1}^{S}
           g_{\ell s}\mb{a}(\phi_{\ell s}),
\label{eq:chnl_mod}
\end{equation}
where $g_{\ell s}\sim\mathcal{CN}(0,1)$, $L$ denotes the number of incoming clusters,
$S$ the number of subpaths in a cluster, $\mb{a}(\phi_{\ell s})$ the steering vector
associated with AoA $\phi_{\ell s}$, and $-\Delta_\ell\leq\phi_{\ell s}-{\phi}_{\ell}
\leq\Delta_\ell$ with ${\phi}_{\ell}$ being the mean AoA of the $\ell$th cluster and
$\Delta_\ell$ the corresponding AS. We consider only the single-cluster case, i.e.,
$L=1$, the extension to multi-cluster channels is touched upon in Section
\ref{section:simulation} where an example is given to validate the extendibility of
the proposed approach.

It is reasonable to assume that the mobile users are relatively far apart (with respect
to the signal wavelength) so that the $k$th uplink SSFC vector is independent of the
$\ell$th vector, $\forall\ell\neq k$. With mutually independent $\tilde{\mb{h}}_k$'s,
each of them can be represented by
\begin{IEEEeqnarray}{rCl}
\mb{h}_k=\mb{\Phi}_{k}^{\frac{1}{2}}\tilde{\mb{h}}_k,
\end{IEEEeqnarray}
where $\mb{\Phi}_{k}$ is the spatial correlation matrix at the BS side with
respect to the $k$th user and $\tilde{\mb{h}}_k \sim \mathcal{CN}(\mb{0}_M,\mb{I}_M)$.
We assume that 
the SSFC matrix $\mb{H}$
remains constant during a pilot sequence period, i.e., the channel's coherence
time is much greater than $T$, while the LSFC $\bm{\beta}$ varies even slower.
Due to the use of orthogonal pilots and the fact that the channel estimation
procedure is identical for all users, we shall omit the user index $k$ henceforth.

\subsection{Analytic Model}
In \cite{Model_based}, two analytic correlated MIMO channel models were proposed.
These models generalize and encompass as special cases, among others, the
Kronecker \cite{Kron_Model}, virtual representation \cite{Sayeed}, and
Weichselberger \cite{Ozcelik} models. They often admit flexible reduced-rank (RR)
representations. For a single-input,
multiple-output (SIMO) channel the analytic MIMO models of \cite{Model_based}
degenerates to the following obvious result.
\begin{lemma}[RR representations]
\label{lemma:reducerank}
The channel vector seen by the $k$th user can be approximated by
\begin{equation}
\mb{h}\approx \mb{Q}_m\mb{c}',~ m\leq M,
\label{C_only}
\end{equation}
or alternately by the phase-modulated (angle-aligned) model
\begin{equation}
\mb{h}\approx \mb{W}(\phi)\mb{Q}_m\mb{c}
\label{RRApp}
\end{equation}
where $\mb{Q}_m =[\mb{q}_1, \mb{q}_2, \cdots, \mb{q}_m] \in \mathbb{C}^{M\times m}$ is a
predetermined (unitary) basis matrix, $\mb{c}'=[c'_1,c'_2,\cdots,c'_m]^T,\mb{c}=[c_1,c_2,
\cdots,c_m]^T \in \mathbb{C}^{m\times 1}$ are the transformed channel vectors for the $k$th
MS-BS link, and $\mb{W}(\phi)$ is diagonal with unit magnitude entries whose phases depend
on a single phase parameter $\phi$. The approximations become exact when $m=M$.
\end{lemma}

The model (\ref{C_only}) seems to be trivial and obvious, after all $\mb{c}'=\mb{Q}_m^H\mb{ h}$
is just a transform-coded version of $\mb{h}$. The not-so-obvious part is that it induces the
decomposition of the spatial correlation matrix
\begin{equation}
\mb{\Phi}=\mathbb{E}\{\mb{h}\mb{h}^H\}
=\mb{Q}_m\mathbb{E}\{\mb{c}'(\mb{c}')^H\}\mb{Q}_m^H
=\sum_k\sum_\ell \overline{c'}_{k\ell}\mb{q}_k\mb{q}_\ell^H,
\label{C_only_1}
\end{equation}
where $\overline{c'}_{k\ell}\defeq\mathbb{E}\{c'_k(c'_\ell)^*\}$.
Similarly, with $\overline{c}_{k\ell}\defeq\mathbb{E}\{c_k c_\ell^*\}$ and $\mb r\defeq\mb{W}^H(\phi)\mb{h}$,
the phase-modulated (PM) channel vector, (\ref{RRApp}) implies
\begin{IEEEeqnarray}{rCl}
\mb{W}^H(\phi)\mb{\Phi}\mb{W}(\phi)&=&\sum_k\sum_\ell \overline{c}_{k\ell}\mb{q}_k\mb{ q}_\ell^H
=\mathbb{E}\{\mb r\mb r^H\}.\label{eq:C_mtx}
\end{IEEEeqnarray}
Both (\ref{C_only_1}) and (\ref{eq:C_mtx}) provide an alternate and perhaps better conceptual
interpretation on our modeling, i.e., we attempt to approximate the spatial correlation matrix
by a two-dimensional (2D) separable transform of a lower-rank matrix. As (\ref{C_only_1}) can
be rewritten as $\mb{Q}_m^H\mb{\Phi}\mb{ Q}_m=\mathbb{E}\{\mb{c}'(\mb{c}')^H\}$, we are in fact
performing a 2D transform coding on $\mb{\Phi}$ to obtain a compressed representation.
The interpretation of (\ref{eq:C_mtx}), on the other hand,
is given in the following remark. Further elaboration from the transform coding perspective
is given in Section \ref{sec:basis} where we discuss the choice of the basis matrix $\mb{Q}_m$.
\begin{remark}[Advantage of PM model]
\label{remark:LPM_merits}
We claim that the rank-reduction capability of the PM model (\ref{RRApp}) can
potentially be better than that of (\ref{C_only}). We argue as following. Let $p(\theta)$ be
the pdf of the AoA and $\mb{u}_i=[u_{ix},u_{iy}]^T$ be the Cartesian coordinates of the $i$th
antenna element. Without loss of generality, we assume that antenna elements $i$ and $j$ lie
on the $y$-axis and the impinging waveform spreads over $[\phi-\Delta,\phi+\Delta]$, where
$\phi$ is the mean AoA and $\Delta$ the AS. Following \cite{SCM_v}, we can show that, for
antenna spacing $d_{ij}\defeq{u}_{iy}-{u}_{jy}$ and small $\Delta$,
\begin{IEEEeqnarray}{rCl}
  [\mb{\Phi}]_{ij}&=&\mathbb{E}\{h_i h_j^*\}=\int_{-\Delta}^{\Delta}p(\theta+\phi)
  e^{-j\frac{2\pi d_{ij}}{\lambda}\sin(\theta+\phi)}\mathrm{d}\theta\notag\\
   &\approx&
   e^{-j\frac{2\pi d_{ij}}{\lambda}\sin\phi}
   \int_{-\Delta}^{\Delta}p(\theta+\phi)e^{-j\frac{2\pi d_{ij}}{\lambda}\sin\theta\cos\phi}\mathrm{d}\theta\notag\\
   &\defeq&e^{-j\frac{2\pi d_{ij}}{\lambda}\sin\phi}\left[\bar{\mb{\Phi}}\right]_{ij}
\label{eq:Phi_ij}
\end{IEEEeqnarray}
where $\left[\bar{\mb{\Phi}}\right]_{ij}$ is real if $p(\theta)$ is symmetric about $\phi$.
This then implies that, for a uniform linear array (ULA) with antenna spacing $\xi$ and
incoming signal wavelength $\lambda$, if we set $[\mb{W}(\phi)]_{ii}=\exp\left[-j2\pi \frac{(i-1)\xi}
{\lambda} \sin\phi\right]$, ${\bf r}$ is linearly PM (LPM) and $\mb{C}\defeq\mathbb{E}\{\mb r\mb r^H\}$
is approximately real. Hence, its {\it real} dimension is effectively only {\it half} that of $\mb{\Phi}$. For other array
geometries, if the diagonal elements of $\mb{W}(\phi)$ form the steering vector associated with
the mean AoA, the same conclusion holds.
\end{remark}
A direct implication of the this remark is that the mean AoA of each user link is extractable if
the associated AS is not large which, as reported in a recent measurement campaign \cite{Measurement2,NewMeasurement},
is likely to be the case 
at the BS. The MSE of a channel estimator based
on (\ref{RRApp}), as to be proved later, is also smaller as a result of a smaller required modeling order.

\section{RR Channel Estimation and Basic Performance Analysis}
\label{sec:SSFC}
\subsection{SSFC Estimation}
\label{section:CE}
We begin with the channel model (\ref{RRApp}) and denote by $\bm{\epsilon }$ the modeling error.
Let ${\gamma}=\sqrt{{\beta}}\|\mb{p}\|^2$ and assume that an $M$-element ULA with antenna
spacing $\xi$ is used. Then, when LSFCs are perfectly known, the post-detection pilot-matched
filter output vector is given by
\begin{IEEEeqnarray}{rCl}
\label{MF output}
    \mb{Y}\mb{p}&=&
    {\sqrt{\beta}}
    \|\mb{p}\|^2\mb{h}+\mb{N}\mb{p}
    \notag\nonumber\\
    &\defeq&\gamma\left(\mb{W}(\phi)\mb{Q}_m\mb{c}+\bm{\epsilon}\right)
     +\mb{N}'.
\end{IEEEeqnarray}
If the modeling error $\bm\epsilon$ is negligible, the ML joint estimator for $\phi,\mb{c}$
is the solution of
\begin{IEEEeqnarray}{rCl}
\label{prob:28}
\min_{\phi,\mb{c}}&&~~\left\|\mb{Y}\mb{p}-\gamma
\mb{W}(\phi)\mb{Q}_{m}\mb{c}\right\|^2\nonumber\\
\mathrm{s.t.}&&~~\mb{W}(\phi)={\Diag}\left(\omega_1(\phi),
\cdots,\omega_M(\phi)\right), \nonumber\\
     &&~~~\omega_i(\phi)=\exp\left(-j2\pi\frac{(i-1)\xi}{\lambda}
     \sin\phi\right)
\end{IEEEeqnarray}
and is of a least-squares (LS) form similar to (17) of \cite{Model_based}. With $\mb{F}_m(\phi)
\defeq\mb{W}(\phi)\mb{Q}_{m}$ and $\mb{A}^\dagger\defeq (\mb{A}^H\mb{A})^{-1}\mb{A}^H$, the optimal
solution to (\ref{prob:28}) is
\begin{IEEEeqnarray}{rCl}
\hat{\phi}
&=&\underset{\phi\in[-\frac{\pi}{2},\frac{\pi}{2}]}{\arg\max}~
\mb{p}^H
\mb{Y}^H\mb{F}_m(\phi)
\mb{F}_m^\dagger(\phi)\mb{Y}\mb{p}
\nonumber \\
&=&\underset{\phi\in[-\frac{\pi}{2},\frac{\pi}{2}]}{\arg\max}~
\left\|\left(\mb{W}(\phi)\mb{Q}_{m}
\right)^H\mb{Y}\mb{p}\right\|^2,
\label{eqn:W_hati}\\
\hat{\mb{c}}
&=&\frac{1}{\gamma}\mb{F}_m^{\dagger}(\hat{\phi})
\mb{Y}\mb{p}
=\frac{1}{\gamma}\mb{Q}_m^H
\mb{W}^H(\hat{\phi})\mb{Y}\mb{p}.
\label{c_hat}
\end{IEEEeqnarray}
When the true LSFC is not available we use its estimate,
$\hat{\gamma}=\hat{\beta}^{\frac{1}{2}}\|\mb{p}\|^2$, and obtain
the SSFC estimate
\begin{equation}
 \hat{\mb{h}}_m=\mb{W}(\hat{\phi})\mb{Q}_{m}
 \hat{\mb{c}},
 \label{h_hat}
\end{equation}
where the subscript $m$ for $\hat{\mb{h}}$ emphasizes that the SSFC vector
estimate is obtained by using a rank-$m$ basis matrix ${\mb Q}_m$. Both (\ref{eqn:W_hati})
and (\ref{c_hat}) require no matrix inversion while $\hat{\phi}$
can be found by a line search. Because of $\mb{W}(\hat{\phi})$, we shall refer to
(\ref{h_hat}) as the LPM-aided (RR) channel estimator.
The one based on (\ref{C_only}), called the regular RR channel estimator, is simply
\begin{equation}
\hat{\mb{h}'}_m=\frac{1}{\gamma}\mb{Q}_{m}\mb{Q}_m^H\mb{Y}\mb{p}.
\label{h'_hat}
\end{equation}
\begin{remark}[Hybrid structure and complexity reduction]
Equation (\ref{c_hat}) can be further rewritten as $\frac{1}{\gamma}\tilde{\mb{Q}}_m^H\mb{Y}\mb{p}$,
where $\tilde{\mb{Q}}_m=\mb{W}(\hat{\phi})\mb{Q}_m$ is an $M \times m$ matrix. If $\mb{Y}$
contains time-domain samples and a unit-modulus basis is used, i.e., $\mb{Q}_m$ consists of
unit-modulus column vectors and so does $\tilde{\mb{Q}}_m$, then $\tilde{\mb{Q}}_m^H\mb{Y}$
can be realized with $m$ instead of $M$ RF chains and the remaining operations can be performed
digitally at the baseband, making possible a hybrid receive array structure. Similarly, (\ref{eqn:W_hati})
also allows such a hybrid implementation in searching for the mean AoA. The implementation
complexity is reduced significantly if $m \ll M$. When there is a limitation on the number
of RF chains, our RR channel estimator is particularly desirable. If $m$ exceeds the available
RF chain number, say $m_0$, a serial realization is possible if the channel's coherent time
is greater than $L=\lceil m/m_0 \rceil$ periods and an $L$-period pilot sequence is sent.
The receiver performs similar operations at each period with a basis matrix that consists of
non-overlapping columns of $\tilde{\mb{Q}}_m$.
\end{remark}

\begin{remark}[A parallel spatial search perspective]
As $\mb{Y}\mb{p}=\gamma \mb{h}$ in a noiseless environment, the operation $\tilde{\mb{Q}}_m^H\mb{Y}
\mb{p}= \tilde{\mb{Q}}_m^H\gamma \mb{h}$ gives us yet another interpretation of our approach:
spatial filtering the channel vector $\mb{h}$ by a filter bank, where each column of $\tilde{\mb{Q}}_m$
represents an $M$-element spatial filter (receive beamforming vector). We shall address the basis
(matrix) selection issue in Section \ref{sec:basis}. From the spatial search viewpoint, (\ref{eqn:W_hati})
tries to search for the mean AoA by finding which spatial filter bank has the maximum squared sum
output magnitude. This seems not to be a very plausible approach after all, the equivalent filter bank is
pointing at multiple directions simultaneously. Using the sum output without looking at individual
filter (beam) outputs is equivalent to determining the AoA by the sum spatial response which may
yield poor spatial resolution. However, reflecting on (\ref{eqn:W_hati}) and the description of (\ref{eq:chnl_mod}) should make it clear that, unlike some known methods which assume a single or multiple
separable, deterministic incoming plane wave, each with a fixed AoA \cite{CE_liter1,CE_liter2,CE_liter5},
we are searching for the {\it mean} AoA of a group of plane waves with random incident angles around the
mean AoA. Unless the AS is much smaller than the beamwidth of a single filter (beam), the channel vector
will induce responses (leakage) in multiple filters. The sum output is resulted from a cluster of steering
directions centered around a mean AoA with an AS determined by $m$, the modeling order.
\label{rem:sp_filter}
\end{remark}

\subsection{Performance Analysis}
Let $\mb{Q}\in\mathcal{U}(M)$, where $\mathcal{U}(M)$ is the
unitary group of degree $M$, i.e., the set of all $M\times M$
unitary matrices. Performing unitary transform $\mb{Q}$ on the LPM
channel vector $\mb r$, we obtain $\mb{u}=\mb{Q}^H\mb{r}=\mb{Q}^H\mb{W}^H(\phi)\mb{h}$
which is equal to $\mb{c}$ in (\ref{RRApp}) when $m=M$.
For a reason to become clear later we refer to the correlation matrix of $\mb{u}$
\begin{IEEEeqnarray}{rCl}
\mb{B}\defeq\mathbb{E}\left\{\mb{Q}^H\mb{r}\mb{r}^H\mb{Q}\right\}
=\mb{Q}^H\mb{C}\mb{Q}
\label{Cu}
\end{IEEEeqnarray}
as the \textit{bias matrix}, where $\mb{Q}=[\mb{q}_1, \mb{q}_2,\cdots, \mb{q}_M]$.
Equation (\ref{Cu}) indicates that the bias matrix is the (separable) 2D unitary transform of the
correlation matrix $\mb{C}$. Let $\sigma_\ell^2\defeq[\mb{B}]_{\ell\ell}$, $1 \leq \ell \leq M$,
be the diagonal entries of $\mb{B}$, i.e., the energy spectrum of the transformed channel vector.
Invoking the identity
\begin{IEEEeqnarray}{rCl}
[\mb{C}]_{\ell\ell}
=\omega_\ell^*(\phi)[\mb{\Phi}]_{\ell\ell}\:\omega_\ell(\phi)
=\left[\mb{\Phi}\right]_{\ell\ell}=1,\notag
\end{IEEEeqnarray}
we obtain
\begin{IEEEeqnarray}{rCl}
\sum_{\ell=1}^M\sigma_\ell^2&=&\tr(\mb{B}) 
=\tr(\mb{Q}^H\mb{C}\mb{Q})=\sum_{\ell=1}^M[\mb{C}]_{\ell\ell} = M,~~~
\end{IEEEeqnarray}
an alternate expression to re-confirm that a unitary transform is energy-preserving,
and $\tr(\mb{B})=\mathbb{E}\{\|\mb{u}\|^2\}=\mathbb{E}\{\|\mb{h} \|^2\}$.
We note that the 2D unitary transform in effect redistributes the energy (or variance) of
$\mb{h}$ (and $\mb{r}$). The LPM operator $\mb{W}^H(\phi)$, besides helping
extract the mean AoA information, reshapes the channel's energy spectrum.

Assuming perfect LSFC information and compensation, we substitute (\ref{c_hat})
into (\ref{h_hat}) to obtain
\begin{IEEEeqnarray}{rCl}
\hat{\mb{h}}_m&=&\frac{1}{\gamma}\mb{W}(\hat{\phi})
\mb{Q}_{m}\mb{Q}_m^H
\mb{W}^H(\hat{\phi})\mb{Y}\mb{p}\label{eq:RR_CE_Yp}\\
&=&\frac{1}{\gamma}\mb{W}(\hat{\phi})
\mb{Q}_{m}\mb{Q}_m^H
\mb{W}^H(\hat{\phi})(\gamma\mb{h}+\mb{N}\mb{p})~~~~\nonumber\\
&\defeq&\mathbb{E}\{\hat{\mb{h}}_m\}+\bm{\nu},
\label{eqn:h_biased}
\end{IEEEeqnarray}
where $\bm{\nu}=\frac{1}{\gamma}\mb{W}(\hat{\phi})
\mb{Q}_{m}\mb{Q}_m^H
\mb{W}^H(\hat{\phi})\mb{N}\mb{p}$.
Denoting by $\MSE(\cdot)$, the MSE of the enclosed channel estimator
we have the decomposition
\begin{IEEEeqnarray}{rCl}
\MSE(\hat{\mb{h}}_m)&=&
\mathbb{E}\left\{\left\|\hat{\mb{h}}_m-\mb{h} \right\|^2\right\}\notag\\
&=&\underbrace{\mathbb{E}\left\{
\left\|\hat{\mb{h}}_m-\mathbb{E}\{\hat{\mb{h}}\}
\right\|^2\right\}}_{\defeq\var\{{\hat{\mb{h}}}_m\}}+
\underbrace{\mathbb{E}\left\{
\left\|\mathbb{E}\{\hat{\mb{h}}_m\}-\mb{h}
\right\|^2\right\}}_{\defeq b({\hat{\mb{h}}}_m)}\vspace{-1.2em}\notag\\
\label{eqn:23}
\end{IEEEeqnarray}
where $\var\{{\hat{\mb{h}}_m}\}$ and $b({\hat{\mb{h}}}_m)$ represent respectively
the variance and bias of the estimator $\hat{\mb{h}}_m$. We prove in \ref{app:pf_MSE} that
\begin{lemma}
\label{thm:MSE}
For the channel estimator based on (\ref{RRApp}) $\hat{\mb{h}}_m$, with $\mb{Q}_m$
comprising $m$ distinct columns of $\mb{Q}$, we have
\begin{IEEEeqnarray}{rCl}
\var\{\hat{\mb{h}}_m\}&=&\frac{m}{\beta\|\mb{p}\|^2}
\label{eqn:var_term}
\end{IEEEeqnarray}
and
\begin{IEEEeqnarray}{rCl}
b(\hat{\mb{h}}_m)
&=&\tr\left(\mb{D}_m\hat{\mb{B}}\right)=\sum_{\ell \in \mathcal{F}^c_m}[\hat{\mb{B}}]_{\ell\ell},
\label{eqn:bias_term}
\end{IEEEeqnarray}
where $\mathcal{F}_{m}$ consists of the indices of $\mb Q$'s column vectors that are included
in $\mb Q_m$, $\mathcal{F}^c_m=\{1,2,\cdots, M\}\setminus\mathcal{F}_m$, $\hat{\mb{B}}$ is the
\textit{bias} matrix defined by (\ref{Cu}) with $\phi$ substituted by $\hat{\phi}$, and $\mb{D}_m
\in \mathbb{R}^{M\times M}$ is diagonal with
\[[\mb{D}_m]_{\ell\ell}=\left\{
                    \begin{array}{ll}
                      1, & \hbox{$\ell\in\mathcal{F}_m^c$;} \\
                      0, & \hbox{otherwise.}
                    \end{array}
                  \right.
\]
\end{lemma}
Using an argument similar to \ref{app:pf_MSE}, we verify that
\begin{corollary}
\label{thm:MSE2}
For the regular SSFC estimator $\hat{\mb{h}}'_m$,
\begin{IEEEeqnarray}{rCl}
\var\{\hat{\mb{h}}'_m\}
&=&\frac{m}{\beta\|\mb{p}\|^2},
\label{eqn:var_term2}\\
b(\hat{\mb{h}}'_m)
&=&\tr\left(\mb{D}_m\tilde{\mb{B}}\right)
\label{eqn:bias_term2}
\end{IEEEeqnarray}
where $\tilde{\mb{B}}=\mb{Q}^H\mb{\Phi}\mb{Q}\in \mathbb{C}^{M\times M}$.
\end{corollary}

Note that $\beta \|\mb{p}\|^2$ is the post-detection SNR per antenna port (i.e., an element of
the vector (\ref{MF output})) and so $\beta \|{\bf p}\|^2/m$ can be regarded as the effective
pre-detection SNR per dimension. Equation (\ref{eqn:var_term}) thus says that the estimator's variance
is inverse proportional to this pre-detection sample SNR and is due entirely to noise $\mb{N}$
but independent of the basis (matrix) used. The estimator's bias, as shown in (\ref{eqn:bias_term})
or (\ref{eqn:bias_term2}), is equal to the sum of the $M-m$ diagonal terms, indexed by $\mathcal{F}^c_m$,
of the positive semidefinite matrix $\hat{\mb{B}}$ or $\tilde{\mb{B}}$ and is caused by the modeling
error $\bm{\epsilon}$. The bias can be minimized by choosing the $m$ columns that span the $m$-dimensional
subspace containing the maximum energy of $\mb{h}$. When the full-rank model, i.e., $m=M$, is used,
we have $\mb{Q}_{m}\mb{Q}_{m}^H=\mb{I}_M$ and $\mb{D}_m=\mb{0}_{M\times M}$. Both estimators become
unbiased and equivalent to the conventional LS estimator \cite{Opt_pilot}
\begin{IEEEeqnarray}{rCl}
\hat{\mb{h}}_M=\hat{\mb{h}}'_M=\frac{1}{\gamma}\mb{Y}\mb{p},
\label{eqn:convetional_LS}
\end{IEEEeqnarray}
which is independent of the basis used,
and give the same performance.

\section{Basis Selection for RR Channel Modeling}
\label{sec:basis}
The MSE performance for channel estimator $\hat{\mb{h}}_m$ based on
(\ref{RRApp}) is given by (\ref{eqn:23})--(\ref{eqn:bias_term}).
We note that 
\begin{IEEEeqnarray}{rCl}
b({\hat{\mb{h}}}_m)=\mathbb{E}\left\{\left\|\hat{\mb{u}}-\hat{\mb{u}}^{(m)}\right\|^2\right\}
\end{IEEEeqnarray}
where $\hat{\mb{u}}^{(m)}$ is the truncated version of $\hat{\mb{u}}=\mb{Q}^H \hat{\mb{r}}
\defeq\mb{Q}^H\mb{W}(\hat{\phi})\mb{h}$ obtained by nulling its $M-m$ elements associated with
$\mathcal{F}^c_m$. As $\var\{\hat{\mb{h}}_m\}$ is independent of the basis matrix, to find a basis
matrix that minimizes the mean squared modeling error for a given modeling order $m$ is equivalent
to find a bias-minimizing basis matrix $\mb{Q}_{\textrm{opt}}$ for which the transformed channel
vector would have the least average power in the truncated $M-m$ coordinates, i.e.,
\begin{IEEEeqnarray}{rCl}
\mb{Q}_{\textrm{opt}}=\arg \underset{\mb{Q}\in\mathcal{U}(M)}{\min}\mathbb{E}\left\{\left\|\mb{ Q}^H\hat{\mb{r}}
-\left(\mb{Q}^H\hat{\mb{r}}\right)^{(m)}\right\|^2\right\}.
\label{prob:minrank}
\end{IEEEeqnarray}
\subsection{A Transform Coding Perspective}
\label{subsec:KLT}
The above problem is related to the notion of \textit{energy compactness} of a unitary
transform in digital waveform coding \cite{KLT,KLT2,DCT}. A unitary transform $\mb{Q}$
on the vector $\mb{r}$ redistributes the vector's energy $\mathbb{E}\{\|\mb{r}\|^2\}$
as $\{\sigma_\ell^2\}_{\ell=1}^M$. 
The resulting energy compactness
is measured by the (subband) coding gain \cite[Ch. 1]{KLT}
\begin{IEEEeqnarray}{rCl}
\tau(\mb{Q})\defeq\frac{\frac{1}{M}\sum_{\ell=1}^M\sigma_\ell^2}
{\left(\prod_{\ell=1}^M\sigma_\ell^2\right)^{\frac{1}{M}}}
=\frac{1}{\left(\prod_{\ell=1}^M\sigma_\ell^2\right)^{\frac{1}{M}}}.
\label{eq:tau}
\end{IEEEeqnarray}
It is known that, under certain regularity conditions, the maximizer of (\ref{eq:tau}) is the Karhunen-Lo\`{e}ve transform
(KLT) $\mb{Q}_{\textrm{KL}}$ which diagonalizes the spatial correlation matrix $\mb{C}$
defined by (\ref{eq:C_mtx}). KLT is also the minimizer of the truncation error
\begin{IEEEeqnarray}{rCl}
\varepsilon_m\defeq{\frac{1}{M}}\mathbb{E}\left\{\left\|\mb{u}-\mb{u}^{(m)}\right\|^2\right\},~\mb{ u}=\mb{Q r}.
\end{IEEEeqnarray}
Hence we have
\begin{lemma}
The best dimension reduction of $\mb h$ is achieved by
setting $\mb{Q}_m$ as the one that consists of the eigenvectors
associated with the largest $m$
eigenvalues of $\mb{C}$. 
The use of KLT basis guarantees
\begin{IEEEeqnarray}{rCl}
  \tau(\mb{Q}_{\textrm{KL}})\big|_{M=M_1}>\tau(\mb{Q}_{\textrm{KL}})\big|_{M=M_2}
\label{mono}
\end{IEEEeqnarray}
for two BS antenna arrays of sizes $M_1 > M_2$.
\label{cor:1}
\end{lemma}
The asymptotic coding gain ($M \rightarrow \infty$) of a large-scale
MIMO system can easily be derived by the transform coding analogy as well.

{\it Lemma \ref{cor:1}} reveals the strictly increasing nature of $\tau(\mb{Q}_{\textrm{KL}})$,
therefore, the asymptotic coding gain is also a coding gain upper bound which increases
as the spectrum roughens, i.e., as the degree of spatial correlation increases.
We conclude that increasing $M$ and/or spatial correlation improves the RR channel
modeling efficiency.

\subsection{Predetermined Bases}
Despite of its modeling efficiency, KLT basis is nonflexible in that it is
channel-dependent and computationally intensive. Prior to the channel estimation,
eigen-decomposition and eigenvalue ordering must be performed on the spatial
correlation matrix $\mb{C}$, which varies from user to user and can be accurately
estimated only if sufficient observations are collected. As a result, we resort
to signal-independent predetermined bases.
The use of predetermined unitary matrices $\mb{Q}_m$ in (\ref{C_only})
or (\ref{RRApp}) avoids the needs of eigen-decomposing $\mb{\Phi}$ or $\mb C$ and ordering the eigenvectors to
obtain the associated KLT basis. Two candidate bases are of special interest to us for
their analogousness to the KLT basis.

\subsubsection{Polynomial Basis \cite{Model_based}}
For a correlated environment having a smooth spatial correlation function,
polynomial basis vectors are adequate to model the spatial variation. An
orthonormal discrete polynomial basis $\mb{Q}=[\mb{q}_1,\cdots,\mb{q}_M]$
is constructed via QR decomposition $\mb{U}=\mb{Q}\mb{R}$, where
$[\mb{U}]_{ij}=(i-1)^{j-1}$, $\forall~ i,j=1,\cdots,M$.
As $\mb{q}_i$'s are arranged in the ascending order of degrees, the RR basis $\mb{Q}_m$
is able to track more rapid variation parts of the channel vector 
by including more higher-degree columns.


\subsubsection{Type-2 Discrete Cosine Transform (DCT)
Basis \cite{DCT}} \label{sec:DCT}
DCT, especially Type-2 DCT (DCT-2 or simply DCT), is widely used for image coding for
its excellent energy compaction capability \cite{DCT,Oppenheim}. For a finite-length
real sequence, its DCT is often energy-concentrated in lower-indexed coefficients as
an $N$-point DCT is equivalent to a $2N$-point DFT of a sequence obtained by cascading
the original $N$-point sequence with its mirrored version. As mentioned in {\it Remark
\ref{remark:LPM_merits}}, the correlation matrix $\mb{C}$ of the LPM channel vector is
approximately real for NB channels with small AS although ${\mb \Phi}$ is complex-valued
in general.
Hence, (\ref{h_hat}) would require a modeling order $m$ smaller than that needed by
$\hat{\mb h}'_m$ if $\phi$ is known or can be reliably estimated.

We want to remark that there are fast algorithms for computing DCT that require
$\mathcal O(M\log_2M)$ operations as FFT does. It is also known that the energy
compaction efficiency of DCT is perhaps the closest to that of KLT \cite{KLT,KLT2}
in the high correlation regime among popular unitary transforms such as Walsh-Hadamard,
Slant, Haar, and discrete Legendre transforms. The last one is a polynomial-based
transform. We compare the efficiencies of the DCT and polynomial bases in representing
NB MIMO channels in Section \ref{section:simulation}.

\section{Further Performance Analysis}
\label{sec:RR_MO}
We now examine the relations of $m$ with other system and channel parameters for a given basis.
\subsection{Effect of Array Size}
Recall that (\ref{eqn:bias_term}) implies that the estimation bias is minimized by the semi-unitary
basis matrix $\left(\mb{Q}_{\textrm{KL}}\right)_{m}$, which consists of the $m$ eigenvectors associated
with the largest $m$ eigenvalues of $\mb{C}$.
Therefore, denote by $\mathcal{U}(M,m)$
the set of all $M \times m$ semi-unitary matrices, we have
\begin{IEEEeqnarray}{rCl}
\underset{\mb{Q}_m\in\mathcal{U}(M,m)}{\arg\min}\MSE(\hat{\mb h}_m)
=\underset{\mb{Q}_m\in\mathcal{U}(M,m)}{\arg\min}b(\hat{\mb h}_m)
=\left(\mb{Q}_{\textrm{KL}}\right)_{m},\notag
\end{IEEEeqnarray}
where the first equality is because $\var\{{\hat{\mb{h}}}_m\}$ is independent of the basis used.
\begin{corollary}[Effect of $M$ {\cite[Ch. 1]{KLT}}]
For two BS antenna arrays of sizes $M_1>M_2$ and a fixed $m$, the biases
associated with the channel estimator (\ref{h_hat}) satisfy
\begin{IEEEeqnarray}{rCl}
\frac{1}{M}b(\hat{\mb h}_m)
\big|_{M=M_1}
&<&\frac{1}{M}b(\hat{\mb h}_m)
\big|_{M=M_2},
\end{IEEEeqnarray}
if $\mb{Q}_m=\left(\mb{Q}_{\textrm{KL}}\right)_{m}$ and $\hat{\phi}=\phi$.
The normalized MSE (NMSE) also has the monotonicity property:
\begin{IEEEeqnarray}{rCl}
\NMSE(\hat{\mb{h}}_m)\big|_{M=M_1}&\defeq&
\frac{1}{M}\MSE(\hat{\mb{h}}_m)\big|_{M=M_1}\notag\\
&<&\NMSE(\hat{\mb{h}}_m)\big|_{M=M_2}.
\end{IEEEeqnarray}
\label{remark:11}
\end{corollary}
\subsection{Optimal Modeling Order} \label{subsec:SSFCanalysis}
To analyze the effect of the modeling order on the MSE performance of the estimator, we define
\begin{definition}
The \textit{optimal} modeling order $m^{\star}$ is the one
that minimizes $\MSE(\hat{\mb{h}}_m)$ when $\hat{\phi}=\phi$ or $\MSE(\hat{\mb{h}}'_m)$.
\label{def:opt_mod_ord}
\end{definition}
\subsubsection{Uncorrelated Channels}
We denote by
\begin{equation}
  \mathrm{\alpha}=\beta\|\mb{p}\|^2/T \label{eq:SNR}
\end{equation}
the average pre-detection pilot SNR per spatial sample at the BS and rewrite the estimator MSE as
\begin{IEEEeqnarray}{rCl}
\MSE(\hat{\mb{h}}_m)
=\frac{m}{\beta\|\mb{p}\|^2}+\tr(\mb{D}_m)
=M-m\left(1-\frac{1/T}{\alpha}\right),\nonumber\\
\label{eqn:27}
\end{IEEEeqnarray}
where the first equality is resulted from the zero spatial correlation
assumption (i.e., $\mb{B}=\mb{I}_M$). The above equation predicts that,
for low SNR ($\SNR< 1/T$), a larger modeling order does not help as one just
tries to fit more parameters to noise-dominant observations. However, if
$\SNR> 1/T$, the MSE improves with increasing modeling order $m$. We conclude
that
\begin{corollary}
The optimal modeling order for uncorrelated channels is $M$ if $\SNR>1/T$.
\end{corollary}

\subsubsection{Correlated Channels}\label{ssec:Ana_CorrCh}
As indicated by (\ref{eqn:23})--(\ref{eqn:bias_term}) and Section \ref{sec:basis},
to minimize the estimation MSE, the energy compactness of $\mb{B}$ has to be considered. For a given
basis, the diagonal terms become more and more nonuniform with increasing spatial correlation.
By {\it Definition \ref{def:opt_mod_ord}} the optimal modeling order $m^\star$ is given by
\begin{equation}
 (m^\star,\mathcal{F}_{m^\star})=\arg\min_{(m,\mathcal{F}_{m})}\left(\frac{m}{T\cdot\SNR}
 +\sum_{\ell \in \mathcal{F}^c_{m}}[\mb{B}]_{\ell\ell}\right).
\end{equation}
The above optimal modeling order is difficult to determine especially when $M \gg 1$. However, the
channel spectrum $[{\mb{B}}]_{\ell\ell}$ or $[\tilde{\mb{B}}]_{\ell\ell}$, as shown in Fig. \ref{fig:sketch_D},
is usually concentrated in a relative narrow band (set of consecutive frequency indices).
Furthermore, the LPM channel spectrum (solid curve), as predicted in {\it Remark \ref{remark:LPM_merits}}
and Section \ref{sec:DCT}), is a lowpass one and whose bandwidth is only about half that of the regular
channel spectrum without LPM, i.e., this figure also verifies the advantage of the model (\ref{RRApp})
over (\ref{C_only}). For both models, the {\it passband} boundaries are located at the bottom of the `waterfall'
regions and most of the channel energy lies within the subspace spanned by those columns of $\mb{Q}$ whose
indices belong to the corresponding {\it passband}. To simplify the task of determining the optimal modeling
order, we henceforth restrict $\mathcal{F}_{m}$ to be of the form $\{k, k+1, \cdots, k+m-1\}$ for some
$1\leq k\leq M-m+1$ and define
\begin{definition}
The bias-reduction modeling order $m_{\eta}$ and the associated subspace index set $\mathcal{F}_{\eta}$ of
efficiency $\eta$, $0<\eta<1$, are
\begin{IEEEeqnarray}{rCl}
(m_{\eta}, \mathcal{F}_{\eta}) = \arg \min \left\{m: \frac{\sum_{\ell\in\mathcal{F}_m}
[\mb{B}]_{\ell\ell}}{M}>\eta\right\}
\label{eq:m_c}
\end{IEEEeqnarray}
where each candidate $\mathcal{F}_m$ consists of $m$ consecutive integers in $\{1,2,\cdots,M\}$. For the
convenience of subsequent discussion, we refer to the subspace $\mathrm{Span}(\mb{q}_i, i \in \mathcal{F}_m)$ as the
{\it dominant subspace} and $\mathcal{F}_m$ the {\it dominant support}.
\label{def:m_c}
\end{definition}
Note that $m_{\eta}$ and $\mathcal{F}_{\eta}$ do not require the SNR information and results in negligible
energy loss if $\eta$ is chosen to be close to $1$, say $0.99$. It is is also much easier to find $m_{\eta}$
than $m^\star$. This definition can be extended to model (\ref{C_only}) directly by replacing $\mb B$ with
$\tilde{\mb B}$. Both modeling orders reflect the energy compactness of the transformed channel spectrum;
they decrease when the spatial correlation increases or equivalently, the dominant rank of $\mb{\Phi}$
decreases. Furthermore, \textit{Definition \ref{def:m_c}} and (\ref{eqn:bias_term}) imply that, for any
$\mathcal{F}_m \supset \mathcal{F}_{\eta}$,
\begin{IEEEeqnarray}{rCl}
  \frac{\sum_{\ell\in\mathcal{F}^c_m} [\mb{B}]_{\ell\ell}}{M} < 1-\eta,~\forall~m \geq m_{\eta},
\end{IEEEeqnarray}
i.e., using a dominant support $\mathcal{F}_m$ containing $\mathcal{F}_{\eta}$ ensures that
$b({\hat{\mb{h}}_m}) <M(1-\eta)$.
\begin{remark}
\label{RM:MSE_obs}
When $\eta\approx 1$, increasing $m$ beyond $m_{\eta}$ cannot improve and, in the variance-dominant region,
$\SNR \ll \frac{m_{\eta}}{MT(1-\eta)}$, even degrades the MSE performance. This is because in this region when
$m \geq m_{\eta}$, the MSE grows approximate linearly with $m$, i.e.,
\begin{IEEEeqnarray*}{rCl}
  \MSE({\hat{\mb{h}}}_m)\approx\var\{\hat{\mb{h}}_m\}
  =\frac{m}{\beta\|\mb{p}\|^2}.
\end{IEEEeqnarray*}
Therefore, in the low SNR ($\SNR$) regime, for $m_{\eta} \leq m_1<m_2$,
\begin{IEEEeqnarray}{rCl}
  \MSE({\hat{\mb{h}}}_{m_1})<\MSE({\hat{\mb{h}}}_{m_2}).
 \label{obs:1}
\end{IEEEeqnarray}
On the other hand, in the bias-dominant region, MSE increases with decreasing $m$ for $m < m_{\eta}$.
This can be easily seen by observing that when $m_1 < m_2 < m_{\eta}$
\begin{IEEEeqnarray}{rCl}
\MSE({\hat{\mb{h}}}_{m_1})-\MSE({\hat{\mb{h}}}_{m_2})
&\approx& \sum_{\ell\in\mathcal{F}^c_{m_1}\setminus\mathcal{F}^c_{m_2}}[\mb{B}]_{\ell\ell}>0.
\label{obs:2}
\end{IEEEeqnarray}
\end{remark}
This remark 
will be verified in Section \ref{section:simulation} via simulation
with some typical correlated channels.

\subsection{SNR Effect on Optimal Modeling Order}
\label{subsec:SNR_eff}
We now investigate the SNR influence on the optimal modeling order. \textit{Lemma \ref{thm:MSE}}
says that, for a given modeling order and $\mathcal{F}_m$ pair, $\MSE({\hat{\mb{h}}}_{m})$ depends
on SNR only, therefore, we have
\begin{corollary}
For a fixed modelling order $m\leq M$,
\begin{equation}
  \MSE({\hat{\mb{h}}}_{m})\big|_{\SNR=\SNR_1}
\leq\MSE({\hat{\mb{h}}}_{m})\big|_{\SNR=\SNR_2}
\end{equation}
if and only if $\SNR_1\geq\SNR_2$.
\end{corollary}
On the other hand, it can be shown that
if $m_1<m_2$, then $\mathcal{F}_{m_1}\subset\mathcal{F}_{m_2}$, where
 \begin{IEEEeqnarray}{rCl}
(m_j,\mathcal{F}_{m_j})&=&\arg\min_{(m,\mathcal{F}_{m})}\MSE({\hat{\mb{h}}}_{m})\big|_{\SNR=\SNR_j},
\label{eqn:m1}
 \end{IEEEeqnarray}
and so
\begin{lemma}
For two average pre-detection SNRs, $\SNR_1<\SNR_2$, the corresponding MSE-minimizing solutions
\begin{IEEEeqnarray}{rCl}
(m_j,\mathcal{F}_{m_j})&=&\arg\min_{(m,\mathcal{F}_{m})}\MSE({\hat{\mb{h}}}_{m})\big|_{\SNR=\SNR_j},
~j=1,2,~
\label{eqn:m2}
\end{IEEEeqnarray}
satisfy $m_1 \leq m_2$.
\label{lemma:MSE_SNR}
\end{lemma}
\begin{proof}
Since
$m_1$ and $m_2$ are respectively the MSE-minimizing orders
corresponding to $\SNR=\SNR_1$ and $\SNR_2$ and
$\MSE({\hat{\mb{h}}}_{m_1})\big|_{\SNR=\SNR_1}
\leq\MSE({\hat{\mb{h}}}_{m_2})\big|_{\SNR=\SNR_1}$,
\begin{IEEEeqnarray}{rCl}
\frac{m_2-m_1}{T\cdot\SNR_1}&\geq&\sum_{\ell\in \mathcal{F}^c_{m_1}}
[\mb{B}]_{\ell\ell}-\sum_{\ell\in \mathcal{F}^c_{m_2}}[\mb{B}]_{\ell\ell}\nonumber\\
&=&\left\{
\begin{array}{rl}
0,~~~~~~~~~~&\text{if }m_1=m_2;\\
\sum_{\ell\in\mathcal{F}^c_{m_1}\setminus\mathcal{F}^c_{m_2}}[\mb{B}]_{\ell\ell},&\text{if }m_1<m_2;\\
-\sum_{\ell\in\mathcal{F}^c_{m_2}\setminus\mathcal{F}^c_{m_1}}[\mb{B}]_{\ell\ell},&\text{if }m_1>m_2.
\end{array}\right.
\notag
\end{IEEEeqnarray}
For $\SNR_2$, $\MSE({\hat{\mb{h}}}_{m_2})\big|_{\SNR=\SNR_2}
\leq\MSE({\hat{\mb{h}}}_{m_1})\big|_{\SNR=\SNR_2}$ gives
\begin{IEEEeqnarray}{rCl}
\frac{m_2-m_1}{T\cdot\SNR_2}&\leq&
\left\{
\begin{array}{rl}
0,~~~~~~~~~~&\text{if }m_1=m_2;\\
\sum_{\ell\in\mathcal{F}^c_{m_1}\setminus\mathcal{F}^c_{m_2}}[\mb{B}]_{\ell\ell},&\text{if }m_1<m_2;\\
-\sum_{\ell\in\mathcal{F}^c_{m_2}\setminus\mathcal{F}^c_{m_1}}[\mb{B}]_{\ell\ell},&\text{if }m_1>m_2.
\end{array}\right.
\notag
\end{IEEEeqnarray}
Thus, it is possible that $m_1=m_2$ and the other two possibilities lead to
\begin{IEEEeqnarray}{rCl}
\left\{
\begin{array}{rl}
\SNR_1\leq\displaystyle{\frac{m_2-m_1}{T\sum_{\ell\in\mathcal{F}^c_{m_1}\setminus\mathcal{F}^c_{m_2}}[\mb{B}]_{\ell\ell}}}
\leq\SNR_2,&\text{if }m_1<m_2;\vspace{.3em}\\
\SNR_2\leq\displaystyle{\frac{m_1-m_2}{T\sum_{\ell\in\mathcal{F}^c_{m_2}\setminus\mathcal{F}^c_{m_1}}[\mb{B}]_{\ell\ell}}}
\leq\SNR_1,&\text{if }m_1>m_2.
\end{array}\right.\notag
\end{IEEEeqnarray}
The case $m_1>m_2$ results in a contradiction to the assumption that $\SNR_1<\SNR_2$, we thus conclude that $m_1\leq m_2$.
\end{proof}

\section{Rank Determination and Mean AoA Estimation}
\label{subsec:MOD}
\begin{algorithm}
    \caption{Iterative Modeling Order Determination (IMOD)}
    \label{alg:IMOD}
    \textbf{Initialization:} Set $\hat{\phi}=0$ and termination criterion $\mb{\chi}_T$.
    \begin{itemize}
    \item[\textbf{1}] (Updating bias matrix) Calculate
\begin{IEEEeqnarray}{rCl}
\hat{\mb{B}}=\mb{Q}^H\mb{W}^H(\hat{\phi})\mb{\Phi}\mb{W}(\hat{\phi})\mb{Q}.
\end{IEEEeqnarray}
\item[\textbf{2}] (Finding dominant subspace) Solve (\ref{eq:m_c}) with $\mb{B}=\hat{\mb{B}}$ to obtain an estimate of
    $({m_\eta},\mathcal{F}_{m_\eta})$ and construct $\mb Q_m$.
    \item[\textbf{3}] (Updating mean AoA) With $m=\hat{m}_\eta$, find $\hat{\phi}$ via (\ref{eqn:W_hati}).
    \item[\textbf{4}] (Recursion) Go to Step \textbf{1}; or terminate and output $\hat{m}_\eta$, $\hat{\phi}$,
and $\mb{Q}_m$ if ${\mb \chi}_T$ is satisfied.
\end{itemize}
\end{algorithm}
\begin{algorithm}
    \caption{Low-Complexity Mean AoA Search}
    \label{alg:Fast_mAoA_Search}
    \textbf{Initialization:} Set step size $\mu$, adjustment ratio $\kappa$, threshold $t$,
    window size $w$, DCT basis $\mb Q$, and termination criterion $\mb{\chi}_T$
\begin{itemize}
    \item[\textbf{1}] (Coarse estimation) Randomly draw a $\hat{\phi}$ from $[-\frac{\pi}{2},
    \frac{\pi}{2}]$. Repeat until $\hat{\phi}$ satisfies $\mb q_1^H\mb{W}^H(\hat{\phi})\mb{\Phi}
    \mb{W}(\hat{\phi})\mb q_1\geq t$.
    \item[\textbf{2}] (Fine adjustment) 
Let $\hat{\phi}^+=\min\{\hat{\phi}+\mu,\frac{\pi}{2}\}$ and $\hat{\phi}^-=\max\{-\frac{\pi}{2},
\hat{\phi}-\mu\}$.
Find
\[\hat{\phi}:=\underset{\theta\in\{\hat{\phi},\hat{\phi}^+,\hat{\phi}^-\}}{\arg\min}~
\sum_{\ell=M-w+1}^M \mb q_\ell^H\mb{W}^H(\theta)\mb{\Phi}\mb{W}(\theta)\mb q_\ell.\]
\item[\textbf{3}] (Recursion) Let $\mu:={\mu}/{\kappa}$ and go to Step {\bf 2}; or terminate
and output $\hat{\phi}$, if $\mb{\chi}_T$ is satisfied.
\end{itemize}
\end{algorithm}
In discussing the optimal or near-optimal modeling order, we assume that the true bias matrix
$\mb{B}=\mb{Q}^H\mb{W}^H(\phi)\mb{\Phi}\mb{W}(\phi)\mb{Q}$ or at least its diagonal entries
(the channel's energy spectrum) is known. As $\phi$ is not known a priori, we propose a joint
rank and mean AoA ($\phi$) estimation method called iterative modeling order determination
(IMOD) algorithm by assuming known $\mb{\Phi}$ whose estimation is addressed later in this
section. As shown in the ensuing section, $\hat{\phi}$ and $\hat{m}_{\eta}$ do not need accurate
estimate of $\mb{\Phi}$, hence neither does $\hat{\bf h}_m$. The IMOD algorithm is summarized as
{\bf Algorithm} \ref{alg:IMOD}.
If we are interested in finding the mean AoA only, an alternate search method we refer to as {\bf Algorithm} 2,
which computes the channel spectrum sequentially, is more efficient. This algorithm is based on the fact
mentioned in Section \ref{ssec:Ana_CorrCh} that the $\phi$-aligned channel spectrum with respect to the DCT basis is a lowpass one
(cf. Fig. \ref{fig:sketch_D}) and the dominant support is the set $\{1,2,\cdots,m\}$.

%
%
%
%

Once the modeling order and mean AoA are determined, we then proceed to estimate the channel vector according to (\ref{h_hat}) or (\ref{h'_hat}), which needs to know
the dominant rank and support
only instead of the complete information about $\mb\Phi$.
Compared with other rank estimation methods \cite{Model_based,MOD_1,MOD_2}, the IMOD algorithm finds ${m_\eta}$ that gives
near-minimum MSE with much less complexity and faster convergence.
In contrast to the singular value decomposition (SVD)-based methods \cite{MOD_1,MOD_2} which are applicable only if the KLT basis is used,
the IMOD method is basis-independent and do not have to perform SVD.
In many cases, the spatial correlation
\begin{IEEEeqnarray}{rCl}
\mb{\Phi}&=&\frac{1}{\gamma^2}\left(\mathbb{E}\left\{\mb{Y}
\mb{p}\left(\mb{Y}\mb{p}\right)^H\right\}-\|\mb{p}\|^2\mb{I}_M\right)
\notag\\&\defeq&
\frac{1}{\gamma^2}\left(\mb{\Psi}-\|\mb{p}\|^2\mb{I}_M\right)
\label{eqn:est_SCM}
\end{IEEEeqnarray}
is not available and has to be estimated prior to applying the IMOD algorithm.
An estimate can be derived by first noticing that the ML estimate for $\mb{\Psi}$ is
\begin{IEEEeqnarray}{rCl}
\hat{\mb{\Psi}}=\frac{1}{J}\sum_{i=1}^J\mb{Y}_i\mb{p}_i
\left(\mb{Y}_i\mb{p}_i\right)^H\label{eq:Psi}
\end{IEEEeqnarray}
where $\mb{Y}_i$ and $\mb{p}_i$ denote respectively the samples received and pilot transmitted
at the $i$th training period, and $J$ is the number of pilot periods used. Substituting (\ref{eq:Psi})
into (\ref{eqn:est_SCM}), we then obtain an estimate of $\mb{\Phi}$. Unlike the KLT-based
estimate which needs accurate spatial correlation knowledge, for a predetermined basis CE, all
one has to know is the index range where most of the channel energy $[\hat{\mb{B}}]_{\ell\ell}$
lies, the detail shape within the range is of not much concern; see (\ref{eq:m_c}). Since the
spatial correlation and LSFCs often vary slowly in time, we can collect enough pilot samples from
multiple training blocks to obtain their reliable estimates. Simulation results reveal that $J=2$
or 3 is sufficient in most cases. The complete RR channel estimation process is described in Fig.
\ref{fig:rank_det_bl_diag}.
\begin{figure*}
    \centering
        \includegraphics[width=6in]{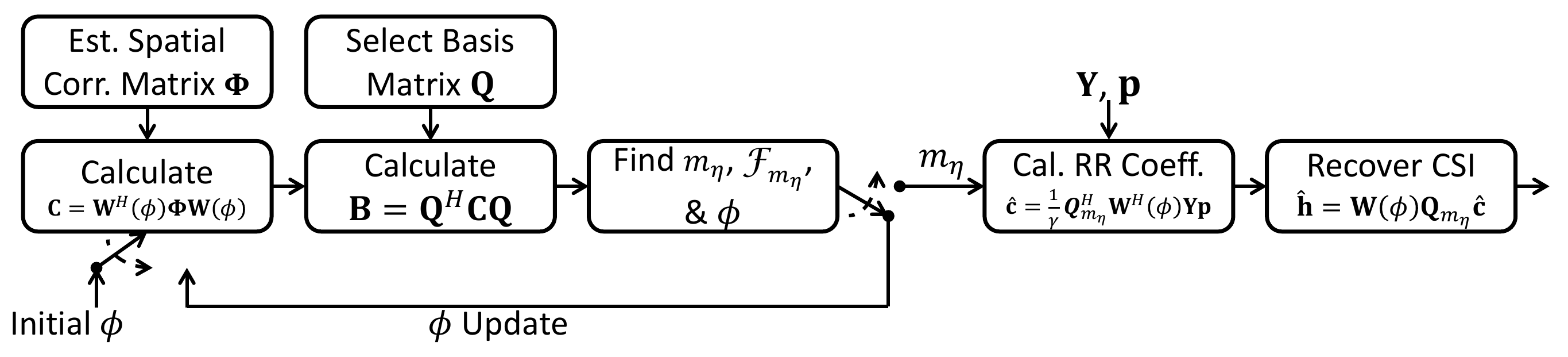}
        \caption{Block diagram of the proposed channel estimation and rank determination procedure.}
 \label{fig:rank_det_bl_diag}
\end{figure*}
\begin{remark}\label{RM:ChCorr_AoA}
It is important at this point to mention two known properties governing the mean AoA, AS, spatial
correlation, and channel rank. The first property is that the spatial correlation decreases
as ULA element spacing and/or the AS $\Delta$ increases \cite{TBW}. The second property further incorporates
the effect of mean AoA $\phi$, stating that the asymptotic normalized channel rank
$\lim_{M\rightarrow \infty} \frac{m}{M}=\rho$ is given by \cite[Thm. 2]{JSDM}
\begin{equation}
\rho=\min\left\{1, \frac{\xi}{\lambda}|\sin(\phi-\Delta)
-\sin(\phi + \Delta)|\right\},
\end{equation}
where $-\pi/2<\phi-\Delta<\phi+\Delta<\pi/2$.
The above equation indicates that for a fixed $\phi$, $\rho$ increases as $\Delta$ becomes larger and;
for a fixed $\Delta$, a $\phi$ that is closer to $0$ leads to a larger $\rho$. In the finite $M$ regime such behaviors
are also verified in \cite{SpCorr_AS&mAoA} and in the following section.
\end{remark}


\section{Numerical Results and Discussion}
\label{section:simulation}
In this section, we investigate the effects of key system design parameters such as
the basis used, modeling order, and channel parameters as SNR (\ref{eq:SNR}), channel correlation,
AS, and mean AoA on the RR channel estimator performance. Unless otherwise specified,
we assume a $100$-element ULA with antenna spacing $0.5\lambda$ and use the standardized
SCM channel model \cite{SCM} with a single cluster that contains $20$ subpaths.
\subsection{AoA Alignment, Bases, and Model Order}
In Section \ref{sec:DCT}, we have mentioned that the estimator $\hat{\mb h}_m$ and
$\hat{\mb h}'_m$ differ in that the former attempts to redistribute the average energy of
an LPM channel vector having an approximately real-valued correlation matrix in a reduced-dimension subspace while the latter has to deal with one that has complex-valued spatial
correlations. We therefore predicted that the LPM based estimator $\hat{\mb h}_m$ performs
better than $\hat{\mb h}'_m$, the one without the LPM operation, in that the former offers
more reduction on rank. Fig. \ref{fig:sketch_D} confirms such a prediction; it shows that one
needs to find only one index ${m_\eta}$ instead of two required by 
$\hat{\mb h}'_m$ which also has a inferior rank reduction capability. Mean AoA estimate is
a by-product of $\hat{\mb h}_m$.
\begin{figure}
    \centering
        \includegraphics[width=3.76in]{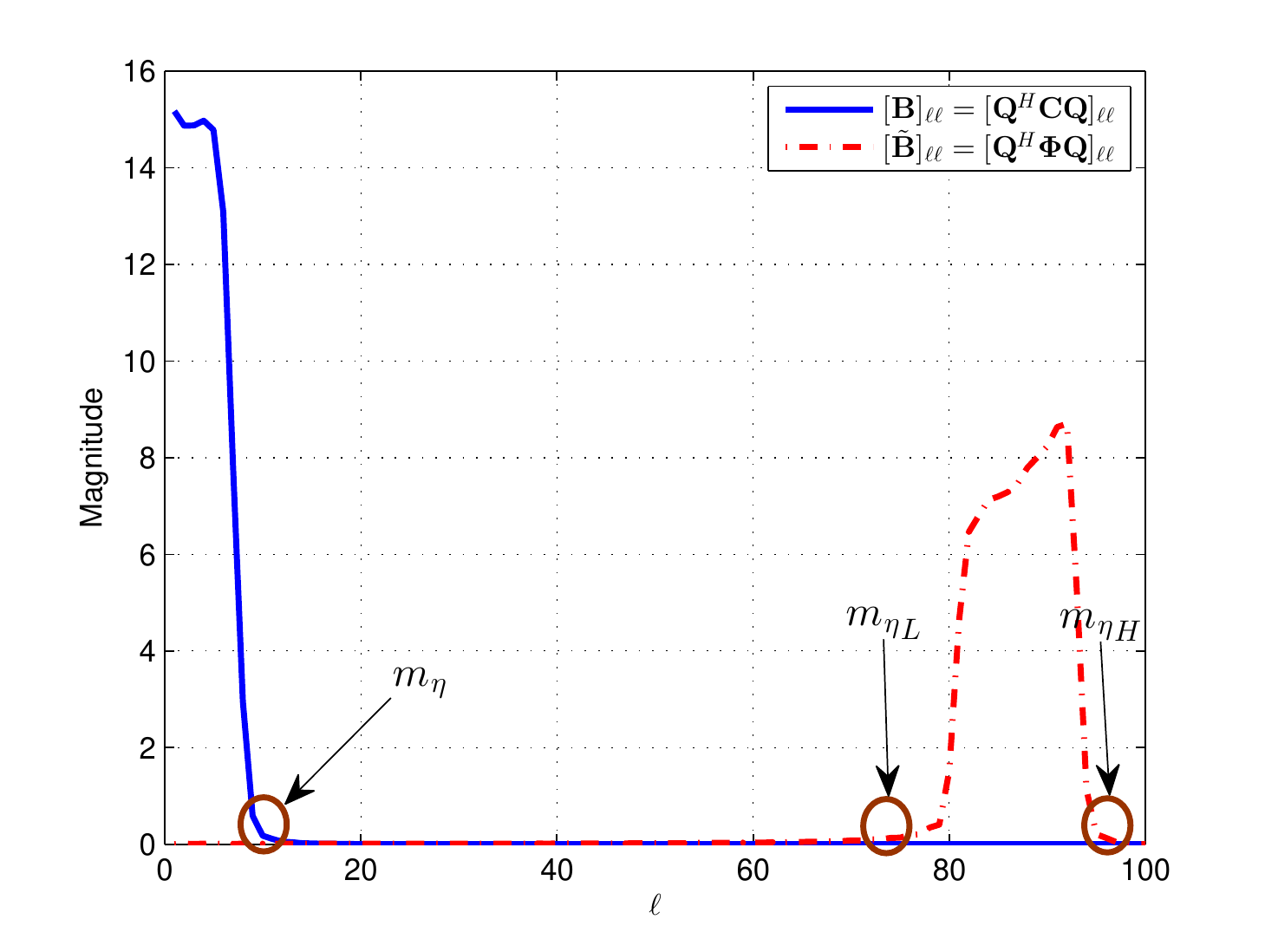}
\caption{Channel Energy spectra (the diagonal terms of bias matrices) with respect to the DCT
basis with (solid curve) and without (dotted curve) mean AoA alignment;
$\xi=\lambda/2$, $M=100$ $\AS=7.2^\circ$, known mean AoA $\phi=60^\circ$.}
 \label{fig:sketch_D}
\end{figure}
\begin{figure}
    \centering
        \includegraphics[width=3.76in]{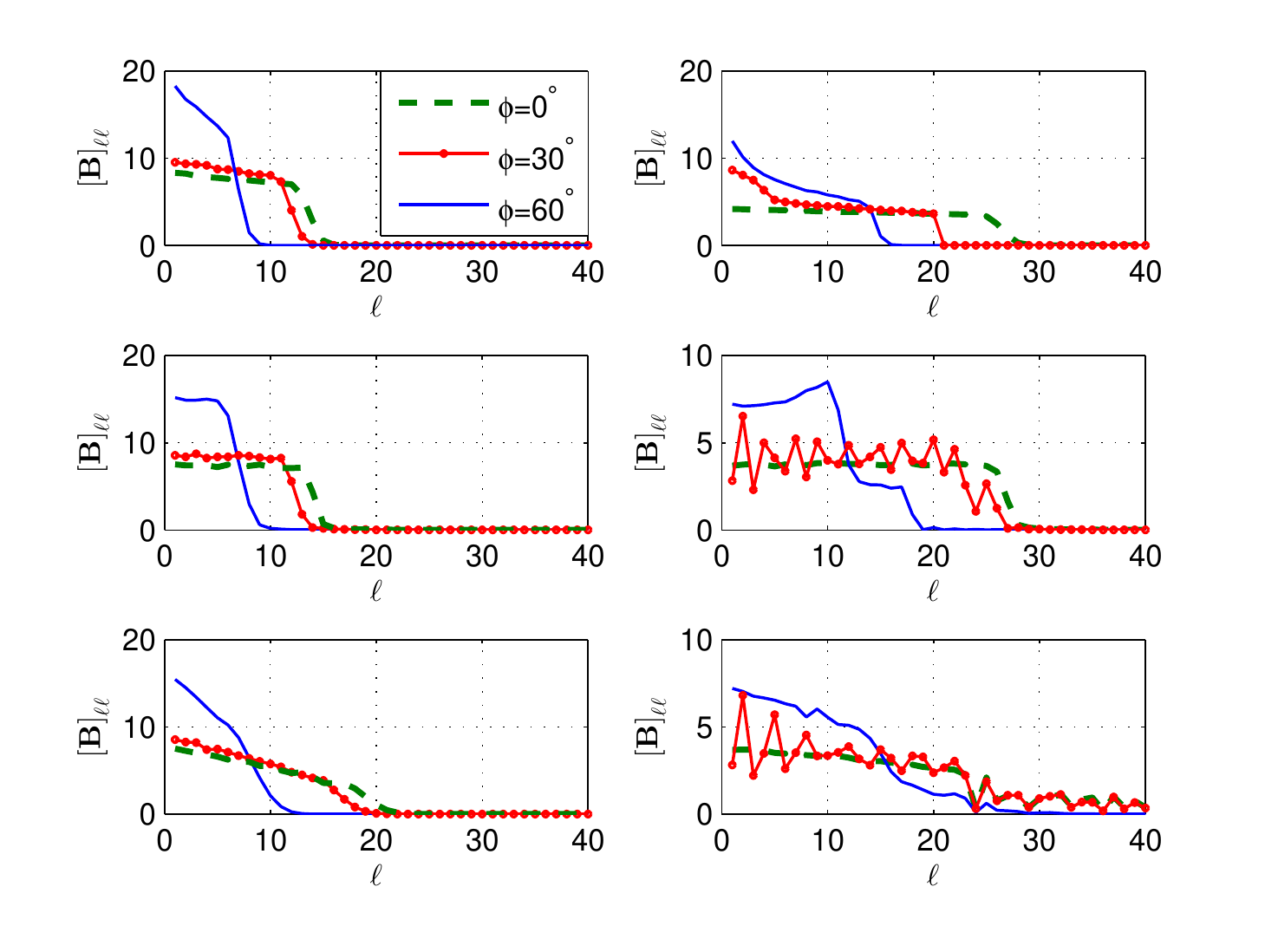}
        \caption{LPM channel energy spectra ($[\mb{B}]_{\ell\ell}$) seen by different bases.
        The left and right columns of sub-figures assume $\AS=7.2^\circ$ and $15^\circ$,
        respectively. The top, middle, and bottom rows are the spectra with respect to KLT, DCT, and
        polynomial bases.} \label{fig:as72aoa60_as15aoa0_KLT_DCT_Poly_Bdiag}
\end{figure}

Fig. \ref{fig:as72aoa60_as15aoa0_KLT_DCT_Poly_Bdiag} shows the energy
spectra of highly ($\AS=7.2^\circ$) and
moderately ($\AS=15^\circ$) correlated channels with respect to the KLT, DCT, and polynomial bases after LPM operation. For all three bases considered,
the channel rank is an increasing function of the AS and a decreasing function of
the mean AoA $\phi$. This trend has been predicted in {\it Remark \ref{RM:ChCorr_AoA}}.
As expected, the KLT basis is indeed the most efficient in that it requires the
least modeling order $m$, among the three bases, to characterize the $M$-dimensional channel.
The DCT and polynomial bases give similar performance
when AS is small but the polynomial basis degrades significantly at $\AS=15^\circ$.

\begin{figure}
    \centering
        \includegraphics[width=3.6in]{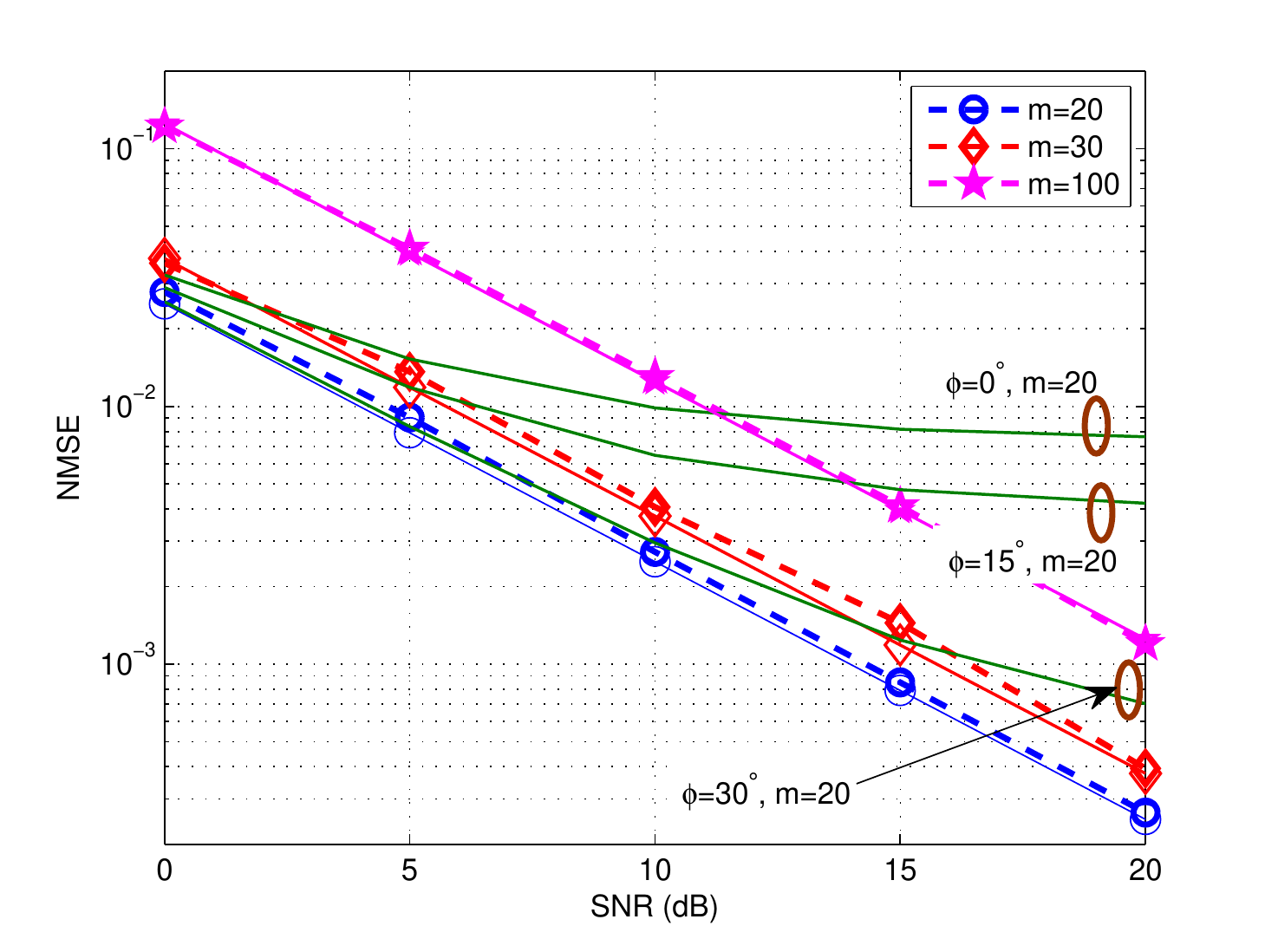}
        \caption{Effect of the modeling order on the MSE performance of the polynomial-based CE $\hat{\mb h}_m$
        with known LSFCs; $\AS=7.2^\circ$, $\phi=60^\circ$. Both theoretical MSE performance with known $\phi$ (solid curves) and simulated counterpart (dashed curves) are shown. Also shown is the theoretical performance as a
        function of $\phi$ when the modeling order is $m=20$.}
 \label{fig:model_order_72_poly}
\end{figure}
\begin{figure}
    \centering
        \includegraphics[width=3.6in]{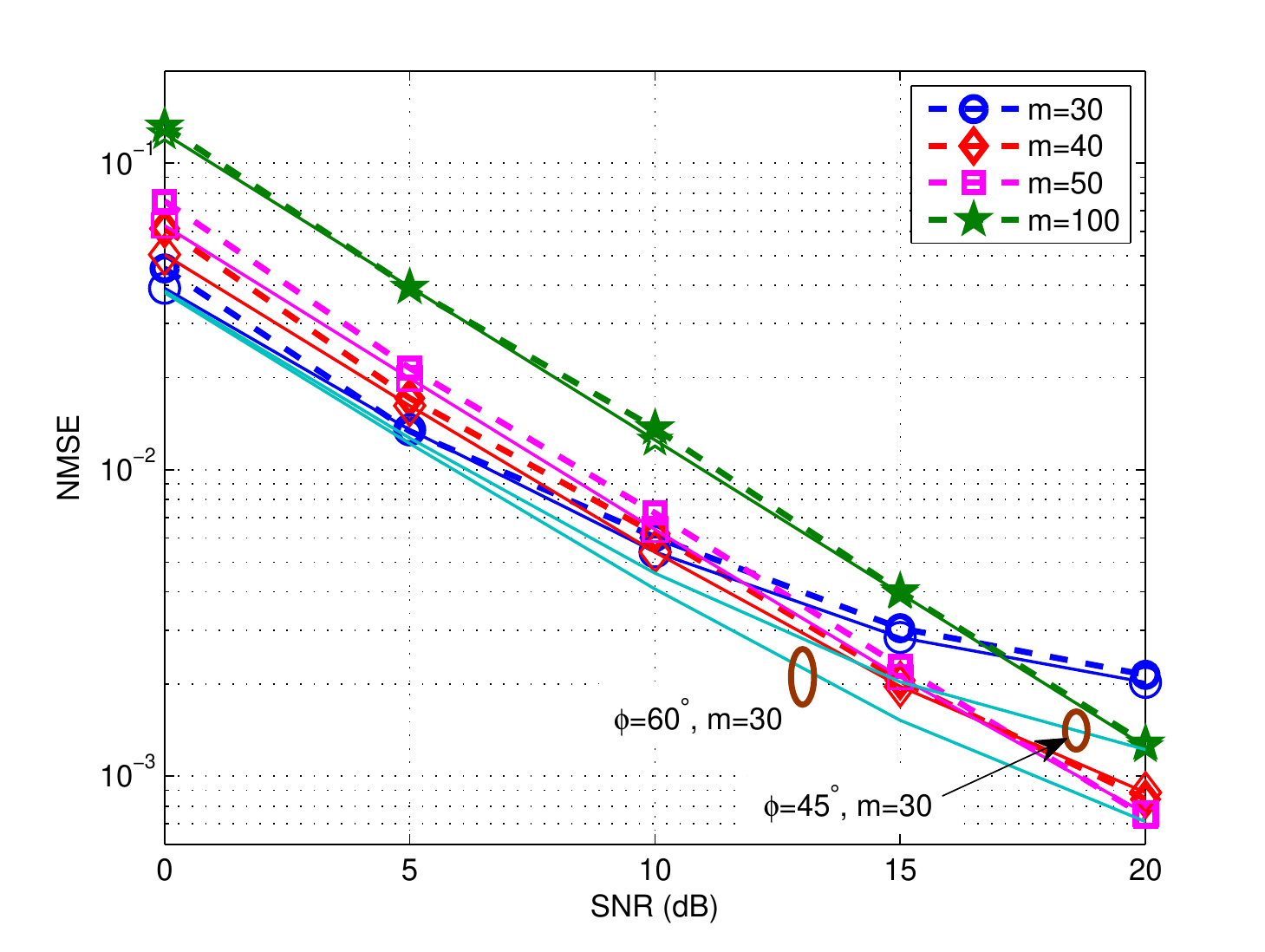}
        \caption{Rank selection effect on the DCT-based CE $\hat{\mb h}_m$'s simulated (dashed) and theoretical
        (solid curves) MSE performance with known LSFCs; $\AS=15^\circ$, $\phi=0^\circ$. The impact of $\phi$ on a
        fixed-$m$ estimator is also shown.}
 \label{fig:model_order_15_DCT}
\end{figure}

The effect of the modeling order on the DCT- and polynomial-based channel estimators' MSE
performance is examined in Figs. \ref{fig:model_order_72_poly} and \ref{fig:model_order_15_DCT}
for $(\AS,\phi)=(7.2^\circ,60^\circ)$ and $(15^\circ,0^\circ)$, respectively. When the AS is small
($7.2^\circ$), the spatial correlation is relatively high (cf. {\it Remark \ref{RM:ChCorr_AoA}})
so we need a small modeling order. Fig. \ref{fig:model_order_72_poly} shows that over-modeling the channel
(as $m$ increases from $20$ to $100$) in fact degrades the MSE performance. But when $\mathrm{AS}$
increases to $15^\circ$, the spatial correlation decreases and the spatial channel roughens, hence
a larger $m$ is called for, especially at higher SNR; see Fig. \ref{fig:model_order_15_DCT}. In both
figures we also show the MSE performance (solid curves without markers) for some selected mean AoAs.
We find that, for a fixed AS and SNR, the optimal modeling order is a decreasing function of $|\phi|$,
which is consistent with what {\it Remark \ref{RM:ChCorr_AoA}} and Fig.
\ref{fig:as72aoa60_as15aoa0_KLT_DCT_Poly_Bdiag} have predicted: the channel correlation increases
with decreasing $|\phi|$.

\begin{figure}
    \centering
        \includegraphics[width=3.76in]{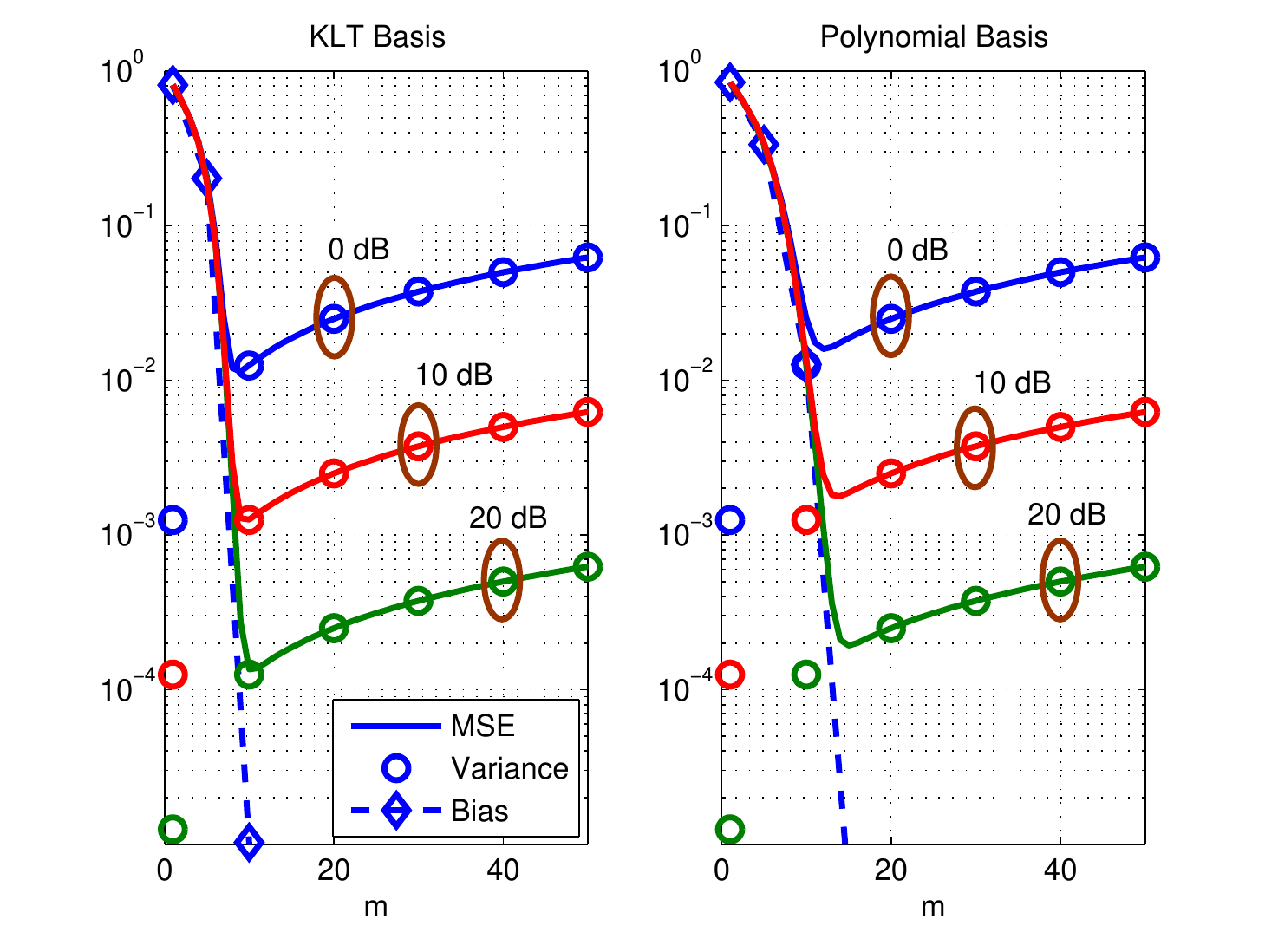}
        \caption{MSE, variance, and bias (all normalized by $M=100$) of the RR channel estimator
        versus modeling order and $\SNR$; $\AS=7.2^\circ$, known mean AoA $\phi=60^\circ$.}
 \label{fig:WC_vs_C_compare_AS72}
\end{figure}
\begin{figure}
    \centering
        \includegraphics[width=3.76in]{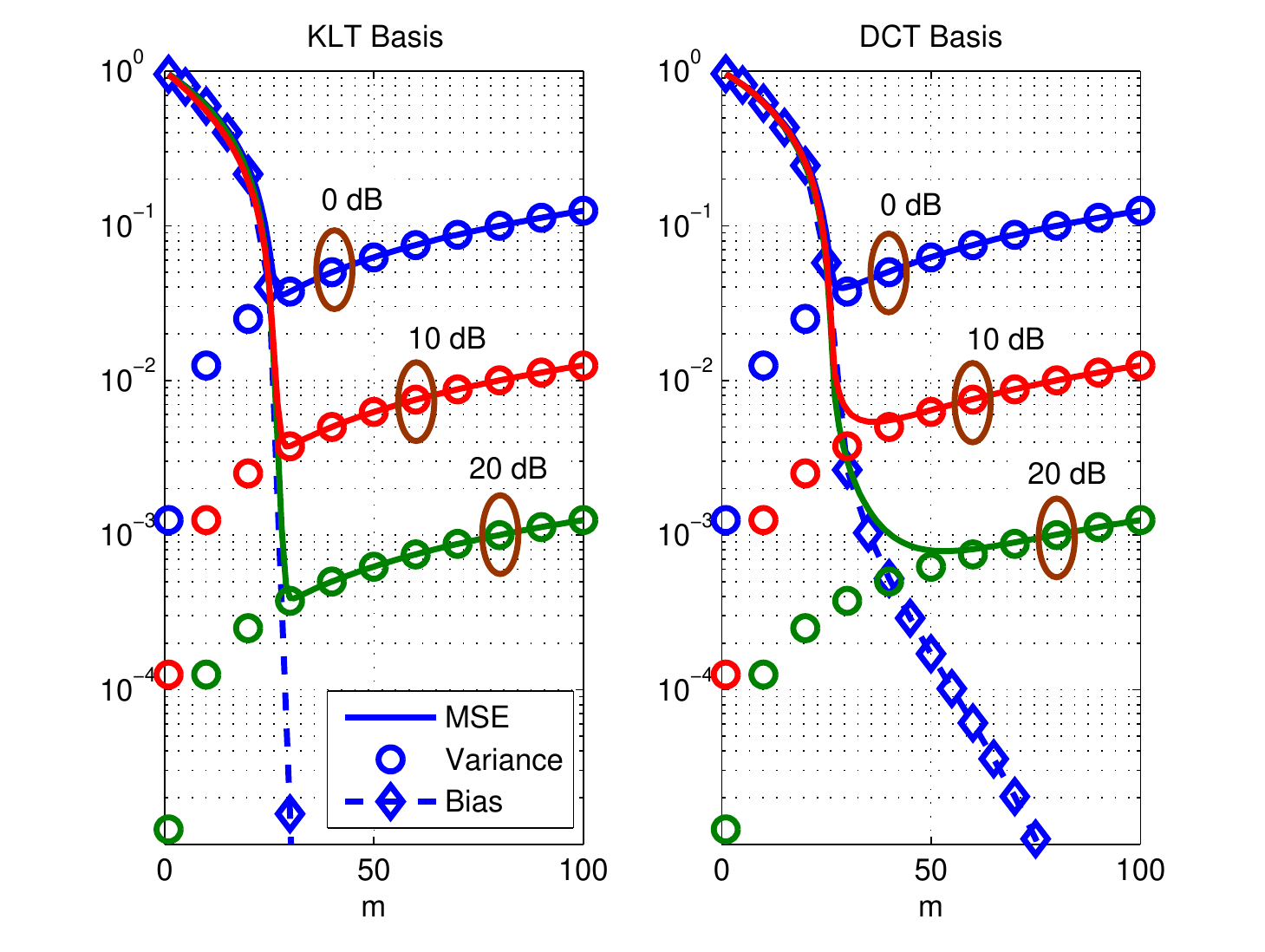}
        \caption{MSE, variance, and bias (all normalized by $M$) behaviors of the RR channel estimator
        versus $m$ and $\SNR$; $\AS=15^\circ$, known mean AoA $\phi=0^\circ$.}
 \label{fig:WC_vs_C_compare_AS15}
\end{figure}

In Figs. \ref{fig:WC_vs_C_compare_AS72} and \ref{fig:WC_vs_C_compare_AS15}, we plot the bias, variance,
and MSE with respect to modeling order $m$ of the estimator (\ref{h_hat}) in a high and moderate spatial correlated channels.
The MSE-minimizing modeling order $m^\star$, for each scenario corresponding to different bases and SNRs
can be found in these two figures. {\it Lemma \ref{lemma:MSE_SNR}} is numerically verified and the associated
$m^\star$ and ${m_\eta}$ are listed in Tables \ref{tab:IMOD_72} and \ref{tab:IMOD_15}. When SNR becomes larger,
MSE becomes more bias-dominant (cf. {\it Remark \ref{RM:MSE_obs}}) and so $\eta$ should be chosen to be closer
to $1$. When $\eta$ is so adaptive to SNR, ${m_\eta}$ is close to ${m}^\star$ and we have near-optimal
MSE performance. As expected, the KLT-based estimator requires the smallest modeling order and achieves the
best MSE performance.

\begin{table}
 \caption{Modeling orders which minimize the MSE of CEs and bias and those derived from 
IMOD algorithm for 
  $(\AS,\phi)=(7.2^\circ,60^\circ)$.}
 \centering
 \tabcolsep 0.05in
 \label{tab:IMOD_72}
\begin{tabular}{|C{.56in}|C{.8in}|C{.8in}|C{.8in}|}
 \hline
\multicolumn{4}{|c|}{KLT Basis} \\ \hline
 $\SNR$&$0$ dB &$10$ dB&$20$ dB \\ \hline
 ${m}^\star$ & $9$&$10$&$10$ \\ \hline
 ${m_\eta}$ ($\phi$, $\mb{\Phi}$ known) & 8\hspace{1in}($\eta=0.99$)&9\hspace{1in}($\eta=0.999$)&10 ($\eta=0.9999$) \\ \hline
 ${\hat{m}_\eta}$ ($\mb{\Phi}$ known)& $8.0$\hspace{1in}($\eta=0.99$)&$9.0$ ($\eta=0.999$)&$10.0$ ($\eta=0.9999$) \\ \hline
 ${\hat{m}_\eta}$ ($\hat{\mb{\Phi}}$, $J=10$)& $10.3$\hspace{1in}($\eta=0.9$)&$19.7$ ($\eta=0.99$)&$20.3$ ($\eta=0.999$) \\ \hline
 \hline
\multicolumn{4}{|c|}{Polynomial Basis} \\ \hline
 $\SNR$&$0$ dB &$10$ dB&$20$ dB \\ \hline
 ${m}^\star$ & $12$&$14$&$15$ \\ \hline
 ${m_\eta}$ ($\phi$, $\mb{\Phi}$ known) & 11\hspace{1in}($\eta=0.99$)&12 ($\eta=0.999$)&14 ($\eta=0.9999$) \\ \hline
 ${\hat{m}_\eta}$ ($\mb{\Phi}$ known)& $13.5$\hspace{1in}($\eta=0.99$)&$16.1$ ($\eta=0.999$)&$17.4$ ($\eta=0.9999$) \\ \hline
 ${\hat{m}_\eta}$ \small{($\hat{\mb{\Phi}}$, $J=10$)}& $9.0$\hspace{1in}($\eta=0.9$)&$12.0$ ($\eta=0.99$)& $13.6$ ($\eta=0.999$) \\ \hline
 \end{tabular}
\end{table}
\begin{table}
 \caption{Modeling orders which minimize the MSE of CEs and bias and those derived from the IMOD
 algorithm for $(\AS,\phi)=(15^\circ,0^\circ)$.}
 \centering
 \tabcolsep 0.05in
 \label{tab:IMOD_15}
\begin{tabular}{|C{.56in}|C{.8in}|C{.8in}|C{.8in}|}
 \hline
\multicolumn{4}{|c|}{KLT Basis} \\ \hline
 $\SNR$&$0$ dB &$10$ dB&$20$ dB \\ \hline
 ${m}^\star$ & $28$&$29$&$31$ \\ \hline
 ${m_\eta}$ ($\phi$, $\mb{\Phi}$ known) & 27\hspace{1in}($\eta=0.99$)&28 ($\eta=0.999$)&30 ($\eta=0.9999$) \\ \hline
 ${\hat{m}_\eta}$ ($\mb{\Phi}$ known)& $27.0$\hspace{1in}($\eta=0.99$)&$28.0$ ($\eta=0.999$)&$30.0$ ($\eta=0.9999$) \\ \hline
 ${\hat{m}_\eta}$ \small{($\hat{\mb{\Phi}}$, $J=10$)}& $10.9$\hspace{1in}($\eta=0.9$)&$20.1$ ($\eta=0.99$)& $21.5$ ($\eta=0.999$) \\ \hline
 \hline
\multicolumn{4}{|c|}{DCT Basis} \\ \hline
 $\SNR$&$0$ dB &$10$ dB&$20$ dB \\ \hline
 ${m}^\star$ & $29$&$37$&$53$ \\ \hline
 ${m_\eta}$ ($\phi$, $\mb{\Phi}$ known) & 27\hspace{1in}($\eta=0.99$)&36 ($\eta=0.999$)&56 ($\eta=0.9999$) \\ \hline
 ${\hat{m}_\eta}$ ($\mb{\Phi}$ known)& $28.3$ ($\eta=0.99$)&$37.7$ ($\eta=0.999$)&$59.1$ ($\eta=0.9999$) \\ \hline
 ${\hat{m}_\eta}$ ($\hat{\mb{\Phi}}$, $J=10$)& $24.8$\hspace{1in}($\eta=0.9$)&$28.6$ ($\eta=0.99$)&$37.2$ ($\eta=0.999$) \\ \hline
 \end{tabular}
\end{table}

\begin{figure}
    \centering
        \includegraphics[width=3.5in]{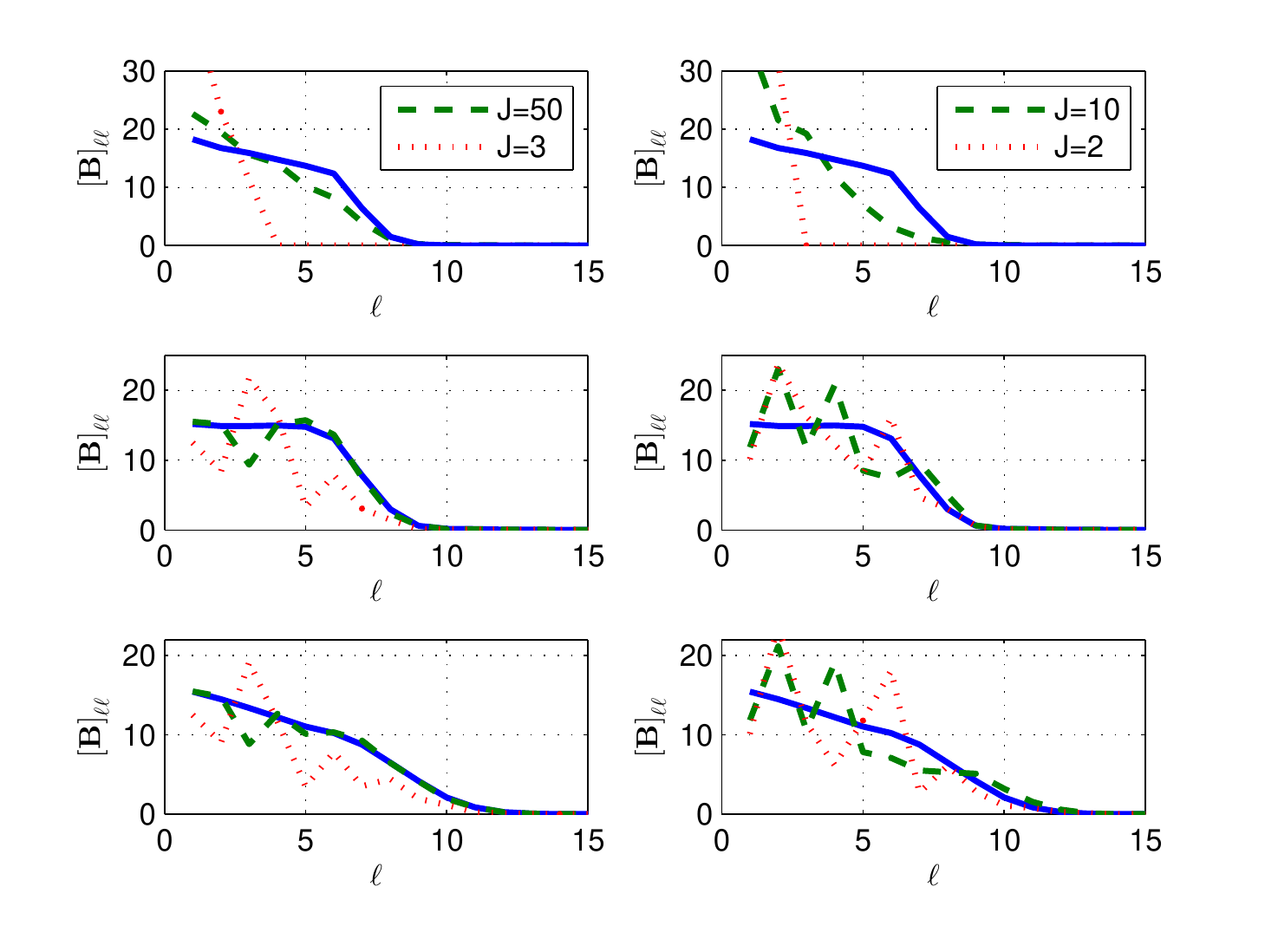}
        \caption{Estimated channel spectra at $\SNR=10$ dB; $\AS=7.2^\circ$, and $\phi=60^\circ$ (known).
        Estimated LPM spectra in the left column of sub-figures are obtained with $50$ or $3$ samples, whereas $10$ or $2$ in the right. The sub-figures in the top, middle, and bottom rows are the spectra with respect to the KLT, DCT, and polynomial bases. The blue solid curves correspond to the results of the true $\mb\Phi$.}
 \label{fig:AS72_sub_50}
\end{figure}

\begin{figure}
    \centering
        \includegraphics[width=3.5in]{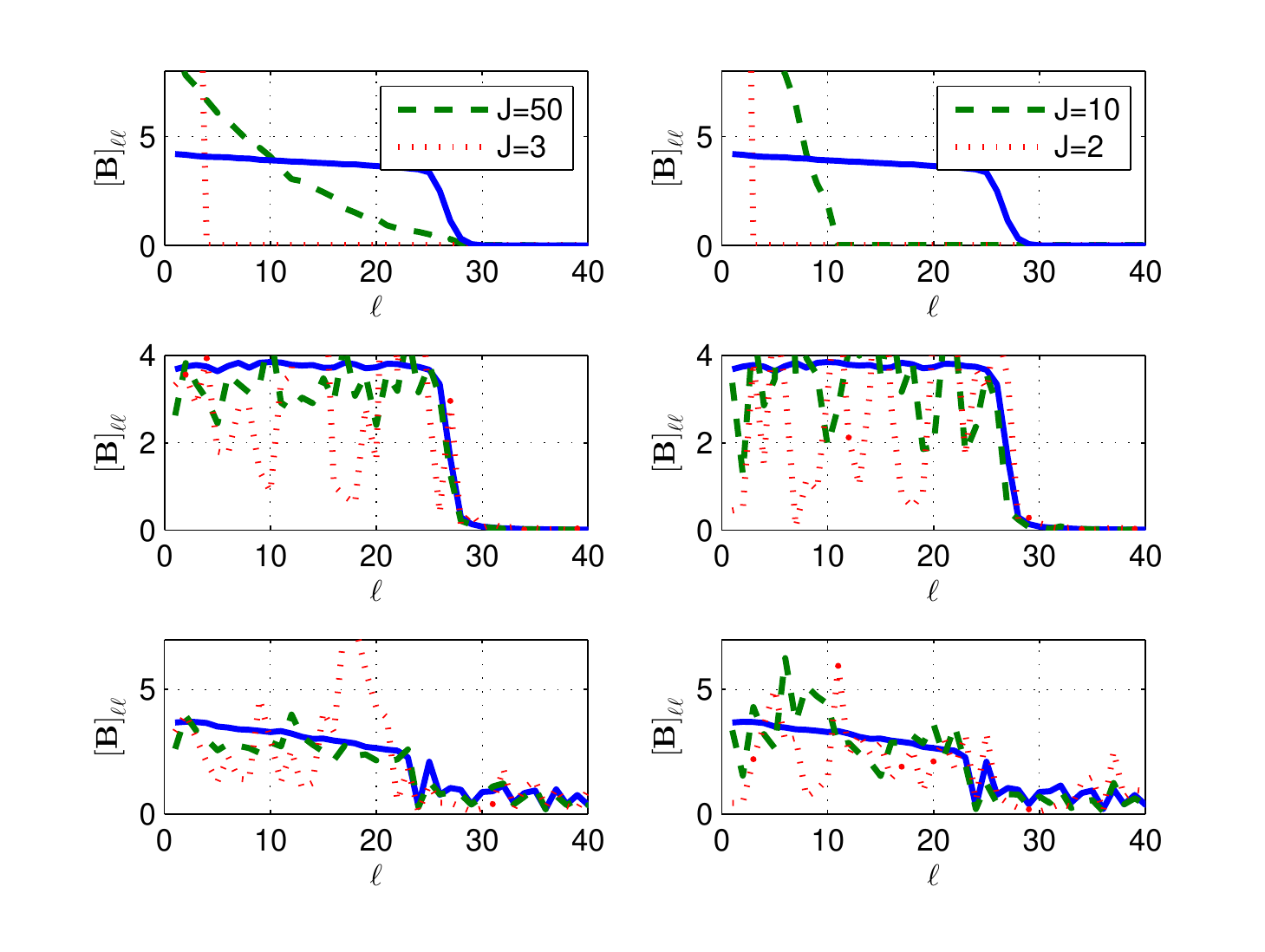}
        \caption{Estimated and true LPM channel spectra at $\SNR=10$ dB; $\AS=15^\circ$, and $\phi=0^\circ$ (known).
        Estimated spectra in the left sub-figures are obtained with $50$ or $3$ samples, whereas $10$ or
        $2$ samples are used elsewhere. The sub-figures in the top, middle, and bottom rows are the spectra with respect to
        the KLT, DCT, and polynomial bases. The blue solid curves are associated with the true $\mb\Phi$.}
 \label{fig:AS15_sub_50}
\end{figure}
\subsection{Spatial Correlation, Receive Beamforming, and Multi-cluster Channels}
Also shown in Table \ref{tab:IMOD_72} and \ref{tab:IMOD_15} are  the average ${\hat{m}_\eta}$'s obtained
via the IMOD algorithm, which usually converges in just a few (no more than five) iterations. A more
realistic scenario when $\mb{\Phi}$ is estimated by substituting $\mb\Psi$ in (\ref{eqn:est_SCM})
with (\ref{eq:Psi}) using $J=10$ periods samples is considered there as well. Figs. \ref{fig:AS72_sub_50}
and \ref{fig:AS15_sub_50} show the corresponding estimated channel spectra. As $\hat{\bm{\Psi}}$ is
obtained by averaging $J$ outer product matrices and has a rank less than or equal to $J$, we cannot
find a proper estimate of ${m_\eta}$ for KLT basis when the true dominant rank exceeds the number of sample periods.
This is a shortcoming of the KLT approach when the associated correlation matrix has to be estimated by
averaging small time-domain samples. This sample-deficient problem exists for other similar rank
determination methods \cite{MOD_1,MOD_2} but is of much less concern for the predetermined basis approach
we have adopted.

\begin{figure}
    \centering
        \includegraphics[width=3.5in]{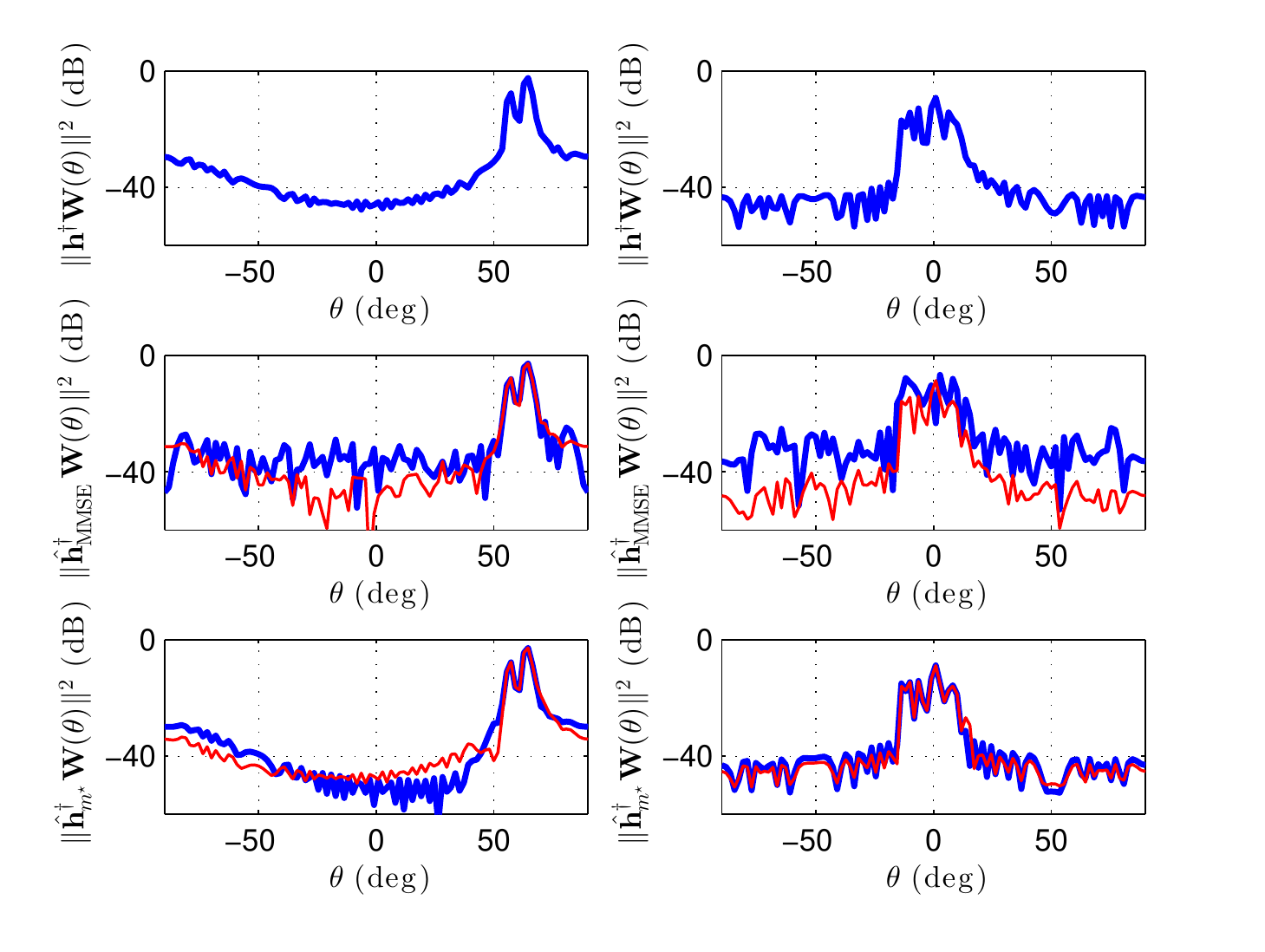}
        \caption{Beam patterns implied by the proposed RR channel estimator for
        two known mean AoAs at $\SNR=0$ dB. The left and right sub-figures assume
        $(\AS,\phi)=(7.2^\circ,60^\circ)$ and $(15^\circ,0^\circ)$ and those in the
        top, middle, and bottom rows are generated by the true channel vector, MMSE
        \cite{MMSE_CE}, and RR estimate (\ref{h_hat}), respectively. Red and blue
        curves represents respectively the results with perfect $\mb\Phi$
        or its estimated counterpart using $J=10$ samples.}
 \label{fig:Ch_BP_AS72_AS15}
\end{figure}

We now verify the effectiveness of our RR channel estimator from the receiver's viewpoint
which is similar to but slightly different from what was discussed in {\it Remark \ref{rem:sp_filter}}.
We notice that a maximum-ratio combining (MRC) receiver would first multiply a received vector ${\mb y}$ by
$\hat{\mb h}^\dag$, the pseudoinverse of the channel estimate $\hat{\mb h}$, before making
hard or soft symbol decisions. This operation is equivalent to receive beamforming. A receive
beamformer which acts as a spatial filter should conceivably match the incoming signal beam
pattern: it should point toward the mean AoA with a (main) beamwidth approximately equal to
the 2-sided AS so that the dominant part of the signal carried by the channel vector is filtered
with minimum distortion while noise or interference from other directions are suppressed. This
prediction is confirmed by Fig. \ref{fig:Ch_BP_AS72_AS15}; see (\ref{eq:RR_CE_Yp}).

Although the MMSE channel estimator \cite{MMSE_CE}
\[
\hat{\mb h}=\mb\Phi\left(\|\mb p\|^2\mb\Phi+\mb I_M\right)^{-1}\left(\mb p^T\otimes\mb I_M\right)
\mathrm{vec}(\mb Y)
\]
yields performance similar to the RR estimators when ${\bf \Phi}$ is perfectly known, it is sensitive
to the correlation matrix estimation error. In contrast, our estimator is more robust as it does not need complete
and accurate information about ${\bf \Phi}$, all it needs is the dominant rank of the channel spectrum
and the associated dominant support.

Finally, to demonstrate that our channel estimator is applicable to multi-cluster channels,
we show two examples in Fig. \ref{fig:2path_DFT_Bdiag_I}.
In the upper sub-figure, the two clusters are separable both spatially, i.e., $[{\phi}_1-\Delta_1,
{\phi}_1+\Delta_1]=[52.8^\circ,67.2^\circ]$ and $[{\phi}_2-\Delta_2,{\phi}_2+\Delta_2]=
[0^\circ,30^\circ]$, and in the DFT domain. As a result, the dominant support consists of two
disjoint sets, $\mathcal{F}_m=[1,{m_\eta}_1]\cup[{m_\eta}_2,{m_\eta}_3]$, so the multi-cluster channel can be estimated via (\ref{C_only}) with $\mb Q_m\defeq
[{\mb q}_\ell]_{\ell\in\mathcal{F}_m}$.
In the lower sub-figure we consider a channel with two overlapped clusters which are inseparable.
By treating this channel as a single-cluster channel, we still can apply our scheme with
dominant support $\mathcal{F}_m=[{m_\eta}_4,{m_\eta}_5]$ to obtain a CE.
\begin{figure}
    \centering
        \includegraphics[width=3.76in]{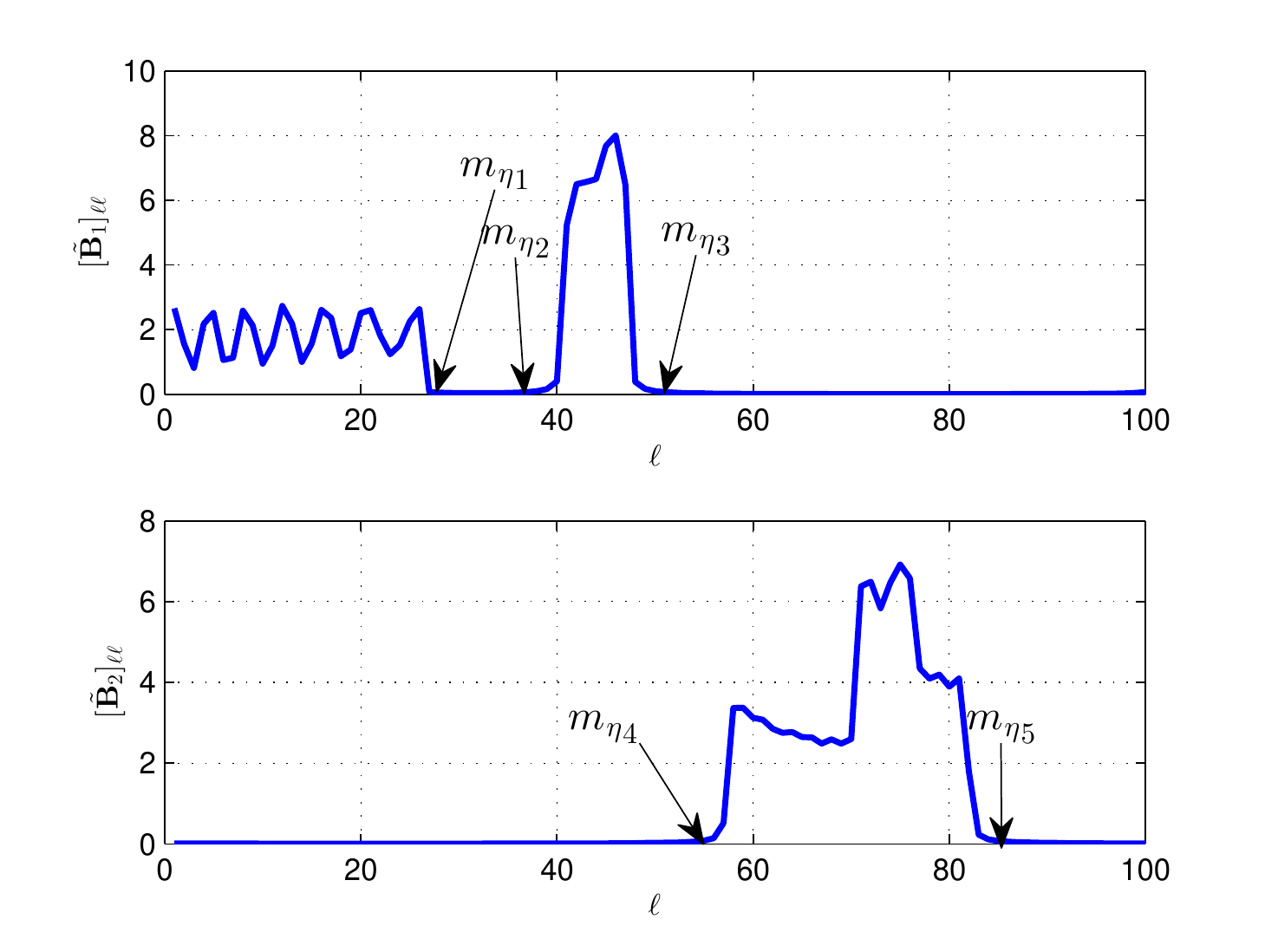}
        \caption{Channel energy spectra of two $2$-cluster channels with respect to the DFT basis.
        (Upper part) $(\Delta_1,{\phi}_1)=(7.2^\circ,60^\circ)$ and $(\Delta_2,{\phi}_2)=(15^\circ,15^\circ)$;
        (Bottom) $(\Delta_1,{\phi}_1)=(7.2^\circ,-30^\circ)$ and $(\Delta_2,{\phi}_2)=(15^\circ,-45^\circ)$.}
 \label{fig:2path_DFT_Bdiag_I}
\end{figure}
\section{Conclusion}
\label{section:conclusion}
In this paper, we extend a model-based RR channel estimator for use in spatially-correlated
massive MIMO systems. The two proposed estimators' rank-reduction nature and looser spatial
correlation information requirement result in performance improvement over conventional LS
and MMSE estimators and lower complexity in post-channel-estimation signal processing. One
of the estimator--the LPM-aided one--obtain a mean AoA estimate as a by-product. The LPM
operation is shown to offer another benefit of further rank reduction when the DCT basis is
employed, resulting in near-optimal energy compaction and performance.

We analyze the impacts of and interrelations among the critical design issues such as basis
selection and rank determination and other system/channel parameters' (e.g., SNR, spatial
correlation, mean AoA and AS) on the estimators' MSE performance. The analytic relation between
the channel's energy spectrum and the estimator's MSE enable us to develop efficient rank
and/or mean AoA estimation methods. Viewing the channel estimation problem from different
perspectives helps casting new insights into the problem of concern. Although we focus our
discussion on single-cluster channels, we briefly demonstrate the feasibility of extending
our approach to more general multi-cluster channels.

\appendices
\renewcommand{\thesection}{Appendix \Alph{section}}
\renewcommand{\theequation}{\Alph{section}.\arabic{equation}}
\renewcommand{\thelemma}{\Alph{section}.\arabic{lemma}}
\setcounter{equation}{0}
\setcounter{lemma}{0}
\section{Proof of Lemma \ref{thm:MSE}}
\label{app:pf_MSE}
\setcounter{equation}{0}
\setcounter{lemma}{0}
In the following, we derive the variance and bias terms of the MSE of $\hat{\mb{h}}_m$ (\ref{RRApp}).
Those of $\hat{\mb{h}}'_m$ for the regular model (\ref{C_only}) can be similarly obtained. We start with
\begin{IEEEeqnarray}{rCl}
\var\{\hat{\mb{h}}_m\}
&=&\mathbb{E}\left\{\left\|\hat{\mb{h}}_m-\mathbb{E}
\{\hat{\mb{h}}\}\right\|^2\right\}~~~~\nonumber\\
&=&\frac{1}{\gamma^2}\mb{p}^H\mathbb{E}\left\{
\mb{N}^H\mb{W}(\hat{\phi})
\mb{Q}_{m}\mb{Q}_m^H
\mb{W}^H(\hat{\phi})\mb{N}\right\}\mb{p}\nonumber\\
&=&\frac{1}{\gamma^2}\mb{p}^H\mathrm{tr}\left(\mb{W}(\hat{\phi})
\mb{Q}_{m}\mb{Q}_m^H
\mb{W}^H(\hat{\phi})\right)\mb{I}_T\mb{p}\notag\\
&{=}&\frac{1}{\gamma^2}\mb{p}^H(m)\mb{p}
=\frac{m}{\beta\|\mb{p}\|^2}
\end{IEEEeqnarray}
where we have invoked the relation
\begin{IEEEeqnarray}{rCl}
\mathbb{E}\left\{\mb{N}^H\mb{X}\mb{N}\right\}
&=&\sum_{i=1}^{M}\sum_{j=1}^{M}x_{ij}\mathbb{E}\left\{\mb{n}_i\mb{n}_j^H
\right\} \nonumber\\
&=&\sum_{i=1}^{M}x_{ii}\mathbb{E}\left\{\mb{n}_i\mb{n}_i^H\right\}
=\mathrm{tr}(\mb{X})\mb{I}_T~~~
\end{IEEEeqnarray}
with white noise $\mb{N}=\left[\mb{n}_1,\cdots,\mb{n}_M\right]^H$
for any square matrix $\mb{X}=[x_{ij}]$. For the bias term, we have
\begin{IEEEeqnarray}{rCl}
b(\hat{\mb{h}}_m)&=&\mathbb{E}\left\{\left
\|\mathbb{E}\{\hat{\mb{h}}_m\}-\mb{h}
\right\|^2\right\} \nonumber\\
&=&\mathbb{E}\left\{\mb{h}^H\left(\mb{W}(\hat{\phi})
\mb{Q}_{m}\mb{Q}_m^H
\mb{W}^H(\hat{\phi})-\mb{I}_M\right)^{\hspace{-.2em}2}
\hspace{-.15em}\mb{h}\right\}\nonumber\\
&=&\tr\left(\left(\mb{W}(\hat{\phi})
\mb{Q}_{m}\mb{Q}_m^H
\mb{W}^H(\hat{\phi})-\mb{I}_M\right)^{\hspace{-.2em}2}
\hspace{-.15em}\mb{\Phi}\right).\vspace{-1em}\nonumber\\
&&\hspace{1.7em}\underbrace{\hspace{13.2em}}_{\defeq\mb{A}_1}\notag
\end{IEEEeqnarray}
Since
\begin{IEEEeqnarray}{rCl}
\mb{Q}_{m}
\mb{Q}_{m}^{H}
&=&
\mb{Q}(\mb{I}_M-\mb{D}_m)
\mb{Q}^{H}
\end{IEEEeqnarray}
where $\mb{D}_m$ 
is
an idempotent matrix, i.e., $\mb{D}_m^2=\mb{D}_m$,
we have $\mb{A}_1=\mb{W}(\hat{\phi})\mb{Q}
\mb{D}_m\mb{Q}^{H}\mb{W}^H(\hat{\phi})$
and
\begin{IEEEeqnarray}{rCl}
b(\hat{\mb{h}}_m)
&=&\tr\left(\mb{W}(\hat{\phi})\mb{Q}\mb{D}_m
\mb{Q}^{H}
\mb{W}^H(\hat{\phi})\mb{\Phi}\right)\label{eqn:46}\nonumber\\
&=&\tr\left(\mb{D}_m\mb{Q}^{H}
\mb{W}^H(\hat{\phi})\mb{\Phi}\mb{W}(\hat{\phi})\mb{Q}\right)\nonumber\\
&=&\tr\left(\mb{D}_m\hat{\mb{B}}\right)=\sum_{\mathcal{F}^c_m}[\hat{\mb{B}}]_{\ell\ell}=\langle \mb{D}_m, \hat{\mb{B}}\rangle_F,~~~~
\end{IEEEeqnarray}
the Frobenius product of $\mb{D}_m$ and $\hat{\mb{B}}$.

\end{document}